\documentclass[twocolumn,preprintnumbers,amsmath,amssymb]{revtex4}

\usepackage{amsmath}
\usepackage{amsfonts}
\usepackage{amssymb}
\usepackage{graphicx}
\usepackage{subfiles}
\usepackage[colorlinks,linkcolor=red,citecolor=green]{hyperref}
\usepackage{subfigure}
\usepackage{graphicx}
\usepackage{dcolumn}
\usepackage{bm}
\usepackage{amssymb,amsmath,color}
\usepackage{booktabs}
\usepackage{amsthm}
\usepackage{dsfont}
\usepackage{tcolorbox}
\usepackage{cases}
\usepackage{physics}
\usepackage{subfiles}

\newcommand\numberthis{\addtocounter{equation}{1}\tag{\theequation}}

\newtheorem{lemma}{Lemma}

\newtheorem{sublemma}{Lemma}[lemma]
\newtheorem{remark}{Remark}

\newtheorem{prop}{Proposition}
\theoremstyle{definition}
\newtheorem{exmp}{Example}
\newcommand{\bc}{\begin{center}}
\newcommand{\ec}{\end{center}}
\newtheorem{thm}{Theorem}

\def\cred{\color{red}}
\def\cbl{\color{blue}}
\def\cbr{\color{brown}}
\def\cgr{\color{green}}
\def\cor{}
\def\blk{\color{black}}

\begin{document}



\def\br{\biggr}
\def\bl{\biggl}
\def\Br{\Biggr}
\def\Bl{\Biggl}
\def\be\begin{equation}
\def\ee{\end{equation}}
\def\bea{\begin{eqnarray}}
\def\eea{\end{eqnarray}}
\def\f{\frac}
\def\n{\nonumber}
\def\l{\label}

\title{An approximation scheme and non-Hermitian re-normalization for description of atom-field system evolution}

\author{B. Ahmadi}
\email{borhan.ahmadi@ug.edu.pl}
\address{International Centre for Theory of Quantum Technologies, University of Gdansk, Jana Bażyńskiego 1A, 80-309 Gdansk, Poland}
\author{R. R. Rodr\'iguez}
\address{International Centre for Theory of Quantum Technologies, University of Gdansk, Jana Bażyńskiego 1A, 80-309 Gdansk, Poland}
\author{R. Alicki}
\address{International Centre for Theory of Quantum Technologies, University of Gdansk, Jana Bażyńskiego 1A, 80-309 Gdansk, Poland}
\author{M. Horodecki}
\address{International Centre for Theory of Quantum Technologies, University of Gdansk, Jana Bażyńskiego 1A, 80-309 Gdansk, Poland}

\begin{abstract}
Interactions between a source of light and atoms are ubiquitous in nature. The study of them is interesting on the fundamental level as well as for applications. They are in the core of Quantum Information Processing tasks and in Quantum Thermodynamics protocols. 
However, even for two-level atom interacting with field in rotating wave approximation there exists no exact solution. This touches a basic problem in quantum field theory, where we can only  calculate the transitions 
in the time asymptotic limits (i.e. minus and plus infinity),
while we are not able to trace the evolution.
In this paper we want to get more insight into the time evolution of a total system of a two-level atom and a continuous-mode quantum field. 
We propose an approximation, which we are able to apply systematically to each order of Dyson expansion, resulting in  greatly simplified formula for the evolution of the combined system at any time. Our tools include a proposed novel, {\it non-Hermitian} re-normalization method.  
As a sanity check, by applying our framework, we derive the known optical Bloch equations.
\end{abstract}

\maketitle
\section{Introduction}

In quantum information processing manipulating two-level systems (qubit), such as trapped ions \cite{Vandersypen,Gulde,Schmidt,Leibfried}, interacting with coherent control laser fields is vitally important for realization of qubit operations since such manipulations constitute the basic components of quantum gates. Within the framework of chemical applications the coherent control field is a fundamental tool in manipulation of ultra-cold molecules, used for instance in laser cooling, photo-disassociation and photo-association, which recently has gained a fast-growing attention \cite{Weyland,Gacesa,Ji,Kon,Green,Kallush,Stevenson,Levin,Carini}. Coherent control fields also play an important role in quantum thermodynamic processes such as short-cuts to adiabaticity technique in which the external field is designed for minimization or maximization of crucial quantities such as time or energy cost of the process \cite{Odelin,Olaya,Sinha,Puebla,Prieto}.

However, the control fields are usually treated classically. Considering the field to be classical causes us to be oblivious to some fundamental quantum effects of the field on the quantum gate. These quantum effects may consist of entanglement of the field with the qubit or spontaneous emission and also the Lamb Shift due to the vacuum effects \cite{Gerry,Breuer}. 
The problem is that even the evolution of two-level atom interacting with light is not exactly solvable, even in rotating wave approximation. In the case of coherent state of the field, one often uses for description of the atom an approximation known as optical Bloch equations \cite{bloch1946nuclear}. The description of dynamics of both atom and field is even more problematic. The exact solution exists just for initial vacuum state of light \cite{Friedrichs,Lee}.

This is related to the general problem in  quantum field  theory where the behavior of the system is examined using the S-matrix (scattering matrix)   which relates the initial state of the system to its final state in asymptotic limits ($t_{initial}\rightarrow-\infty$ and $t_{final}\rightarrow\infty$) \cite{Shankar}. Thus, little is known about the system during the intermediate times. 


The aim of this paper is to propose an approximation scheme, that allows for analytic examination of  time evolution of a total system of  atom and  field for all times, in all orders of Dyson series. As a validation of our approach, we then show that the obtained formulas for the approximate evolution reproduce the well-known  optical Bloch equations.  
 
We consider the interaction of a two-level atom with a continuous-mode quantum field by taking all modes of the field directly into consideration. Using a novel re-normalization method we will derive a greatly simplified formulae for the evolution of the whole system at any time such that it depends only on normally-ordered creation and annihilation operators of the field, and two parameters: the decay rate of the atom and the Lamb Shift in atomic frequency. We will finally show that 
the optical Bloch equations
\cite{bloch1946nuclear} directly and rigorously emerge from our formalism by successively applying our approximation without any further assumption. Apart from the proposed approximation scheme, we heavily base on the "dissipative re-normalization" that we introduce in this paper, which can be of separate interest. Namely, we shift from interaction to self Hamiltonian a non-Hermitian operator, which Hermitian part corresponds to standard re-normalization (related to Lamb-shift) while anti-Hermitian part corresponds to decay rate. 

Our article is organized as follows. In Section \ref{II}, we briefly explain the approximation used in our calculations as well as the re-normalization scheme. In Section \ref{III}, we re-normalize the Hamiltonian of a continuous mode electric field acting on a two level atom. Later, in section \ref{IV}, we provide justifications for our approximation making use of simpler models such as Friedrichs-Lee \cite{Friedrichs,Lee,lonigroquantum}. We compute the explicit evolution of the re-normalized S-propagator (Theorem \ref{thm1}, get the differential equation for S and check our results with particular examples in section \ref{V}.  In section \ref{VI}, we illustrate our results from Theorem \ref{thm1} by means of two opposite examples: the most classical case in which the laser starts in a coherent state, and when its initial state is a continuous superposition of a photon in different modes. For the former, we obtained the well-known Optical Bloch equations which serves us as a sanity check. Finally, in section \ref{VII}, we present the conclusions of our paper and pose possible future research directions.
\section{Outline of the results}\label{II}

As schematically depicted in Fig. \ref{Fig0} we will consider the interaction of a two-level atom with a quantized continuous-mode laser field. The Hamiltonian will be re-normalized which gives rise to a non-Hermitian free Hamiltonian for the atom and consequently a non-Hermitian interaction Hamiltonian
\begin{equation}\label{nf}
H_{A_r}=\hbar(\omega_0+\delta\omega-i\gamma)|1\rangle\langle 1|,
\end{equation}
\begin{equation}\label{n-int}
H_{I_r}=H_{I}-\hbar(\delta\omega-i\gamma)|1\rangle\langle 1|,
\end{equation}
where $\omega_0$ is the atomic frequency and $\gamma$ the decay rate of the atom and $\delta\omega$ the ``Lamb shift'' in atomic frequency. From now on, we use the notation in which $\hbar=1$. The role of performing re-normalization is the following. The Hermitian part - is more or less standard. It is just to put the infinities (or cut-off dependent terms), emerging from commutation relations at each order of Dyson expansion, into single  parameter - the Lamb shift,
whose value we assume can be taken from higher order theory - i.e. quantum electrodynamics. The non-Hermitian part is a novel trick, that in a sense separates spontaneous emission from the evolution, and allows to arrive at an simplified form of the evolution. At the moment it is a technical tool, whose deeper interpretation is still awaiting. Due to the emerged non-Hermitian interaction Hamiltonian the evolution of the atom-field in the interaction picture becomes \textit{non-unitary}, i.e.,
\begin{equation}
|\psi_{I_r,AF}(t)\rangle=e^{-iH_{I_r}t}|\psi_{AF}(0)\rangle.
\label{nu}
\end{equation}
We then suppose that the probability amplitude of the atom in the excited state is of exponential decay form, namely, we shall assume that the long time behavior of the survival amplitude is just an exponential decay. As will be seen in the following in order for this assumption to hold we must apply an approximation which, in turn, leads to a solvable model reproducing the standard second order approximation to the decay rate $\gamma$ and the Lamb Shift $\delta\omega$. In fact, within the approximation regime, for long times,  the probability amplitude of the atom in the excited state decays exponentially.

Applying the approximation the contribution of all terms, which appear after normal ordering of the creation and annihilation operators of the field in Dyson series, is pushed to the re-normalization and the re-normalized evolution deals only with normally ordered terms of creation and annihilation operators of the field, being therefore in principle exactly solvable by means of coherent states basis.
\begin{figure}[h]
\centering
\includegraphics[width=7cm]{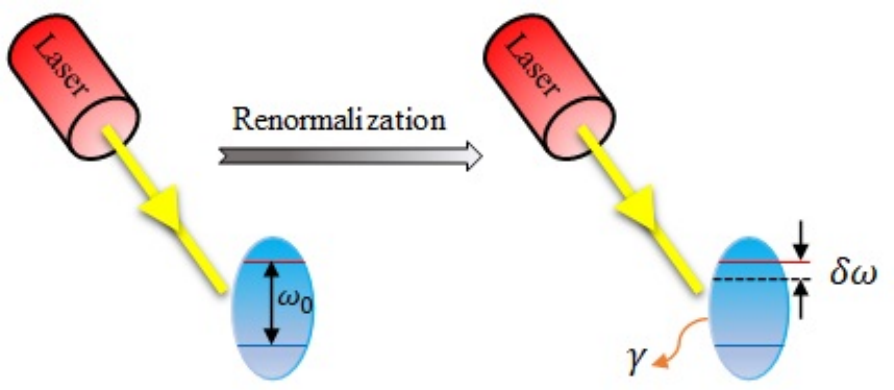}
\caption{Interaction of a two-level atom with a quantized continuous-mode laser field. After re-normalization the re-normalized free Hamiltonian of the atom reads $H_{0_r}=(\omega_0+\delta\omega-i\gamma)|1\rangle\langle 1|$ where $\omega_0$ is the atomic frequency, $\gamma$ the decay rate due to spontaneous emission, $\delta\omega$ the Lamb Shift and $|1\rangle$ the excited state of the atom. As is seen the free Hamiltonian is non-Hermitian. The re-normalized interaction Hamiltonian becomes $H_{I_r}=H_{I}-(\delta\omega-i\gamma)|1\rangle\langle 1|$ which is non-Hermitian.}
\label{Fig0}
\end{figure}
\section{Dissipative re-normalization}\label{renormalization}\label{III}
In this section we first introduce the general form of the Hamiltonian of the interaction of a two-level atom with a continuous-mode field. Then by adding and subtracting two terms from the Hamiltonian we will re-normalize the Hamiltonian and using this re-normalized Hamiltonian we will formulate the \textit{non-unitary} evolution of the atom-field in the interaction picture. For an atom interacting with a continuous-mode field the Hamiltonian, in the rotating wave approximation, reads ($\hbar=1$)
\begin{equation}\label{H0F}
H = H_0 + H_{I},
\end{equation}
where $H_0 =H_A + H_F$ with
\begin{align}\label{H0&HA}
H_A \equiv \omega_0|1\rangle\langle 1|, \ \qquad
H_F \equiv \sum_{\lambda}\int d^3\vec{k} \omega_{\vec{k}\lambda}a^{\dagger}_{\vec{k}\lambda}a_{\vec{k}\lambda},\numberthis
\end{align}
$H_A$ and $H_F$ being the Hamiltonians of the atom and the field respectively and
\begin{equation}
 H_{I} = \sum_{\lambda}\int d^3\vec{k}\, [f(\omega_{\vec{k}\lambda})\sigma^+  a_{\vec{k}\lambda} + f^*(\omega_{\vec{k}\lambda})  \sigma^- a^{\dagger}_{\vec{k}\lambda}],
\label{Hint}
\end{equation}
where the modes are labelled by a continuous wave-vector $\vec{k}$ and a polarization label $\lambda=\pm1$ and $a_{\vec{k}\lambda}$ ($a^{\dagger}_{\vec{k}\lambda}$) are the field annihilation (creation) operators of mode $(\vec{k},\lambda)$, $\sigma^-$ ($\sigma^+$) the atom lowering (raising) operator and
\begin{equation}
f(\omega_{\vec{k}\lambda})=i \sqrt{\dfrac{\omega_{\vec{k}\lambda}}{2(2 \pi)^3\varepsilon_0}}\ \vec{e}_{\vec{k}\lambda}\cdot\vec{D}
\end{equation}
is the coupling constant in which $\vec{e}_{\vec{k}\lambda}$ is the electric field unit vector and $\vec{D}$ the atomic dipole moment vector \cite{mandel1995optical}.  In the following, for ease of calculations, $\int d^3\vec{k}$ will be denoted by the symbol $\int dk$ and the polarization index $\lambda$ is also dropped. The Hamiltonian of the atom-field given in Eq. \eqref{H0F} can be decomposed into two non-Hermitian  re-normalized parts
\begin{equation}
H = H_{0_r} + H_{I_r} ,
\label{HrF}
\end{equation}
where $H_{0_r} = H_{A_r} + H_F$ with
\begin{align}\label{H0r}
H_{A_r} \equiv\Omega |1\rangle\langle 1|, \ \qquad 
H_F \equiv \int d k\, \omega_k a^{\dagger}_ka_k,\numberthis
\end{align}
and 
\begin{equation}\label{Hamr1}
 H_{I_r} = \int d k\, [f(\omega_k)\sigma^+  a_k + f^*(\omega_k)  \sigma^- a^{\dagger}_k]+ib|1\rangle\langle 1|,
\end{equation}
with
\begin{equation}\label{re-constants}
b=i\delta\omega+\gamma,\ \quad \Omega=\omega_{A}-i\gamma,\ \quad \omega_{A}=\omega_{0}+\delta\omega, 
\end{equation}
where $\gamma$ is the decay rate and $\delta \omega$ the Lamb shift (see Fig. \ref{Fig0} for more detail). The interaction Hamiltonian in the interaction picture reads
\begin{eqnarray}
\label{L10}\nonumber
\tilde{H}_{I_r}(t)&=&e^{iH_{0_r}t}H_{I_r}e^{-iH_{0_r}t}\\ \nonumber
&=&\int dk\ [f(\omega_k,t)\sigma^+a_k+f'(\omega_k,t)\sigma^-a^\dagger_k]+ib|1\rangle\langle1|,\\
\end{eqnarray}
where
\begin{equation}\label{f}
f(\omega,t)=f(\omega)e^{-i(\omega-\omega_A+i\gamma)t}, f'(\omega,t)=f^*(\omega)e^{i(\omega-\omega_A+i\gamma)t}.
\end{equation}
It is also convenient to use the notation
\begin{eqnarray}\label{AA}
A(t)\equiv\int dk\ f(\omega_k,t)a_{k},& \nonumber \\
A'(t)\equiv\int dk\ f'(\omega_k,t)a^\dagger_{k},&
\end{eqnarray}
so that
\begin{equation}\label{eq:HI_A}
\tilde{H}_{I_r}(t)= \sigma^+ A(t) + \sigma^- A'(t)+ib|1\rangle\langle1|.
\end{equation}
The operators $A, A'$ satisfy the commutation relations (inherited
from canonical commutation relations):
\begin{equation}\label{CL1}
[A(t),A'(s)]=F^r(t-s),
\end{equation}
where
\begin{align}\label{Frt}
F^r(t)= e^{i(\omega_A-i\gamma)t}F(t),\quad F(t)\equiv\int dk\ e^{-i\omega_kt}|f(\omega_k)|^2.
\end{align}
For the (immediate) proof, see Lemma \ref{lem:commuting_A} in Appendix. As will be seen, in the following, using this re-normalized Hamiltonian the evolution of the S-propagator of the atom-field, in the weak coupling limit, will be remarkably simplified.
\section{Approximation scheme}\label{IV}
Here we will propose an approximation scheme in which we shall assume that the long time behavior of the survival amplitude is just exponential decay. As will be seen below, this approximation is equivalent to the following approximation 
\begin{equation}\label{approximation}
F^r(t)\approx b\delta(t),\ \quad t>0.
\end{equation}
Applying this approximation allows us to remove the contribution of \textit{all} non-normally ordered terms from the evolution of the S-propagator of the atom-field. In fact, the contribution of non-normally ordered terms will be contained just in two terms - the Lamb shift, and the decay rate. We shall first motivate the approximation in a simpler model, that is exactly solvable - Friedrichs-Lee model \cite{Lee,Friedrichs,lonigroquantum}. Then in the next section the full model will be presented.
\subsection{ Friedrichs-Lee model}
We will consider Friedrichs-Lee model to justify our approximation \cite{Lee,Friedrichs,lonigroquantum}. Consider the field to be initially in the vacuum state. Therefore the evolution of the atom-field leaves invariant the sector of Hilbert space spanned by  the vectors: $|e\rangle \equiv |1\rangle\otimes|\{0\}\rangle$,
$|f\rangle \equiv \int dk|0\rangle\otimes f(\omega_k)a_k^{\dagger}|\{0\}\rangle$ where $|1\rangle$ ($|0\rangle$) is the excited (ground) state of the atom and $|\{0\}\rangle$ $\left(|\{1_k\}\rangle\right)$ the vacuum state of the field (the state with one photon in the $k$ mode with an arbitrary $f(\omega_k)$). Thus the interaction Hamiltonian, given in Eq. \eqref{Hint}, restricted to this sector reads
\begin{equation}
 H_{I} = |e\rangle\langle f|+|f\rangle\langle e|, \quad  |f\rangle = \int dk\ f(\omega_k) |0\rangle|\{1_k\}\rangle,
\label{Ham1}
\end{equation}
where $\{\omega_k\}$ denotes the complete set of frequencies that specify the states in each excited mode of the field. We will apply the approximation
\begin{equation}\label{amp-app}
\langle e| e^{-iHt} |e\rangle \simeq  e^{(-i\omega_A - \gamma)t}
\end{equation}
for longer times in the following calculations where $H$ is the Hamiltonian of the total system. The evolution in the interaction picture $U_{I}(t) = e^{iH_0 t} e^{-iH t}$ satisfies the integral equation
\begin{equation}\label{eqW}
U_{I}(t) = \mathbb{I} - i \int_0^t ds\ e^{iH_0 s} (|e\rangle\langle f| + |f\rangle\langle e|)e^{-iH_0 s}U_{I}(s).
\end{equation}
where $H$ and $H_0$ are defined in Eqs. \eqref{H0F} and \eqref{H0&HA}, respectively, (notice that the polarization index $\lambda$ has been dropped for ease of calculations). Here we define the matrix elements
\begin{eqnarray}\label{eqAB}
&K(t) = \langle e|U_{I}(t) |e\rangle ,\qquad  M(t) = \langle f_t|U_{I}(t) |e\rangle
\end{eqnarray}
and the ket 
\begin{equation}\label{ft}
|f_t\rangle=\int dk\ e^{i\omega_kt}f(\omega_k)|0\rangle|\{1_k\}\rangle.
\end{equation}
Now inserting Eq. \eqref{eqW} into Eq. \eqref{eqAB} one obtains
\begin{eqnarray}\label{solA}
&K(t) = 1 - i \int_0^t ds\ e^{i\omega_0 s} M(s) ,\ \qquad & \nonumber \\ 
&M(t) = -i \int_0^t ds\ e^{-i\omega_0 s}F(t-s) K(s),&
\end{eqnarray}
where $F(t)$ is defined in Eq. \eqref{Frt}. Using the definition of the Laplace transform 
\begin{equation}\label{Lap}
\mathcal{F}(z) = \int_0^{\infty} dt\ e^{-z t} F(t),
\end{equation}
and its properties one can readily transform Eq. \eqref{solA} into
\begin{equation}\label{solLap}
\mathcal{K}(z) = \frac{1}{z} - i \frac{1}{z}\mathcal{M}(z - i\omega_0) , \quad  \mathcal{M}(z) = -i \mathcal{F}(z)\mathcal{K}(z + i\omega_0),
\end{equation}
which finally gives
\begin{equation}\label{solfin}
\mathcal{K}(z) = \frac{1}{z + \mathcal{F}(z - i\omega_0)}.
\end{equation}
\subsection{Dissipative re-normalization scheme in Friedrichs-Lee model}
The (non-unitary) re-normalized interaction picture evolution $U_{I_r}(t) \equiv e^{iH_{0_r}t}e^{-iHt}$ reads
\begin{eqnarray}\label{eqWr}\nonumber
U_{I_r}(t) &=& 1-i\int_0^tds\ e^{iH_{0_r}s}\bigr[|e\rangle\langle f| + |f\rangle\langle e|\\
&+&ib|1\rangle\langle 1|\bigr]e^{-iH_{0_r} s}U_{I_r}(s),
\end{eqnarray}
where $H$ and $H_{0_r}$ are defined in Eqs. \eqref{HrF} and \eqref{H0r}, respectively. The relevant matrix element can be expressed in terms of re-normalized interaction picture
\begin{equation}\label{rendyn}
\langle e|e^{-iHt}|e\rangle = e^{(-i\omega_A - \gamma)t} \langle e|U_{I_r}(t)|e\rangle
\end{equation}
We introduce the notation
\begin{equation}\label{not}
K^r(t) = \langle e|U_{I_r}(t)|e\rangle, \quad  M^r(t)  = \langle f_t|U_{I_r}(t)|e\rangle,
\end{equation}
where $|f_t\rangle$ is defined as before. Inserting Eq. \eqref{eqWr} into Eq. \eqref{not} one obtains
\begin{equation}\label{solAr}
K^r(t) = 1 - i \int_0^t ds\ e^{i\Omega s} M^r(s)\, ds  + i(\delta\omega - i\gamma) \int_0^t ds\ K^r(s),
\end{equation}
\begin{equation}\label{solBr}
M^r(t) = -i \int_0^t ds\ e^{-i\Omega s} F(t-s) K^r(s).
\end{equation}
Applying the Laplace transform we get 
\begin{eqnarray}\label{solLap}\nonumber
&\mathcal{K}^r(z) = \frac{1}{z} - i \frac{1}{z}\mathcal{M}^r(z - i\Omega) + i (\delta\omega - i\gamma)\frac{1}{z}\mathcal{K}^r(z)  ,& \qquad \\ 
&\mathcal{M}^r(z)= -i \mathcal{F}(z)\mathcal{K}^r(z + i\Omega),&
\end{eqnarray}
which finally gives
\begin{equation}\label{solfin}
\mathcal{K}^r(z) = \frac{1}{z + \mathcal{F}(z- i\omega_A - \gamma) - i\delta\omega -\gamma}.
\end{equation}
Thus, for example,  the evolution of the probability amplitude that the atom stays excited is given by 
\begin{equation}\label{rendyn}
\langle e|e^{-iHt}|e\rangle =  e^{(-i\omega_A - \gamma)t}K^r(t),
\end{equation}
where $K^r(t)$ is the inverse Laplace transform of
$\mathcal{K}^r(z)$ from  Eq. \eqref{solfin}.
\subsection{Justification of the approximation}\label{subC}

Now we will propose the approximation, mentioned above in Eq. \eqref{approximation}, that will later be carried out to all orders of Dyson series in the general case. According to Eq. \eqref{amp-app} this approximation means that for long times we have $K^r(t) \simeq 1$. In order to translate it into the picture of Laplace transform, we can use the Tauberian theorem
\begin{equation}\label{taub}
\mathcal{Y}(z) \simeq z^{-n}\,\mathrm{as}\,  z\to 0 \iff  Y(t) \simeq \frac{n}{\Gamma(n+1)} t^{(n-1)}
\end{equation}
as $t \to \infty$, where $\Gamma(x)$ is the Gamma function. Putting $n=1$ we conclude that $K^r(t)\simeq 1$ for long times if and only if $\mathcal{K}^r(z)\simeq 1/z$ for small $z$. In Eq. \eqref{solfin} this is equivalent to the condition
\begin{equation}\label{Markov}
\mathcal{F}(z-i\omega_A-\gamma) \simeq \gamma + i\delta\omega 
\end{equation}
or equivalently
\begin{equation}\label{laplace-app}
  \mathcal{F}^r(z)\simeq b.  
\end{equation}
Taking the inverse Laplace transform of Eq. \eqref{Markov} our weak coupling approximation \eqref{amp-app}, in time domain, becomes
\begin{equation}\label{+approximation}
F^r(t)\approx b\delta(t),\ \quad t>0.
\end{equation}
And for negative times we have (see Appendix \ref{proof-approximation})
\begin{equation}\label{-approximation}
F^r(t)\approx b^*\delta(t),\ \quad t<0.
\end{equation}

So far this approximation was considered for long times. 
Now, we propose to allow for {\it substitution}
\begin{equation}\label{+approximationreplacement}
    F^r(t) \to b \delta(t), \quad \text{for}\quad  t>0,
\end{equation}
\begin{equation}\label{-approximationreplacement}
    F^r(t) \to b^* \delta(t) \quad \text{for}\quad  t<0,
\end{equation}
under time integrals. We expect this approximation is valid for relatively small coupling, but stronger than typical weak coupling scenario as encountered in quantum optics, allowing. The substitution works due to oscillatory behavior of $F^r(t)$, and it is close in spirit to secular approximation. However, it is much less harmful leaving room for non-Markovian effects. In Appendix \ref{DeltaPlot} we also show that in typical time integrals, used in derivation of our main result, our approximation still holds with a good accuracy for coupling weak enough. 
\blk

In the rest of the paper we shall apply the approximation to get simplified equations of motion for spin boson model.
In particular, we shall validate the resulting equations by showing that they  reproduce the well known Optical Bloch equations.

We finish this section, by showing that the obtained  values of $\gamma$ and $\delta \omega$ are consistent with 
the assumption, that for times long enough, we have Markovian evolution. 
\blk
In order to compute the values of $\gamma$ and $\delta \omega$ we substitute Eq. \eqref{Frt} into Eq. \eqref{laplace-app} and using dispersion relation $\omega=k c $ we get (with speed of light c=1) 
\begin{equation}\label{Fform}
 \mathcal{F}^r(z)  =  \int_0^{\infty} dt\ e^{-(z_r -\gamma)t} e^{i(\omega_A-\omega-z_i)t}\int d\omega\  4\pi\omega^2 |f(\omega)|^2,
\end{equation}
where $z_r\equiv \mathfrak{R}(z)$, $z_i\equiv \mathfrak{I}(z)$. As we can clearly see, the only way for this integral not to be divergent for $t\rightarrow \infty$ is that $z_r -\gamma > 0$. Considering that, we have
\begin{equation}\label{final-laplace}
 \mathcal{F}^r(z) = \int d\omega\ 4\pi\omega^2 \frac{|f(\omega)|^2}{z+i(\omega-\omega_A)-\gamma}.
\end{equation} 
For small $z$ (and consequently $\gamma$) we have
\begin{eqnarray}\label{DP}\nonumber
\dfrac{1}{z-i\omega_A-\gamma+i\omega}&=&\dfrac{z_r-\gamma+i(\omega_A-\omega-z_i)}{|z_r-\gamma|^2+|\omega_A-\omega-z_i|^2}\\\nonumber
&=&\dfrac{z_r-\gamma}{|z_r-\gamma|^2+|\omega_A-\omega-z_i|^2}\\\nonumber
&+&i\dfrac{\omega_A-\omega-z_i}{|z_r-\gamma|^2+|\omega_A-\omega-z_i|^2} \\
&\approx&\pi\delta(\omega_A-\omega)+i\mathcal{P}(\dfrac{1}{\omega_A-\omega}),
\end{eqnarray}
where  $\mathcal{P}$ is the Cauchy principal value. Now substituting Eqs. \eqref{laplace-app} and \eqref{DP} into Eq. \eqref{final-laplace} one gets
\begin{equation}\label{DP2}
\gamma +  i\delta\omega + O(z)\approx4\pi^2\omega^2_A|f(\omega_A)|^2+i\mathcal{P}\int d\omega\ \frac{4\pi\omega^2}{\omega_A-\omega} |f(\omega)|^2,
\end{equation}
which gives
\begin{equation}\label{DP3}
\gamma\approx4\pi^2\omega^2_A|f(\omega_A)|^2,\ \quad \delta\omega\approx\mathcal{P}\int d\omega\ \frac{4\pi\omega^2}{\omega_A-\omega} |f(\omega)|^2.
\end{equation}

We have just obtained standard Markovian decay rate, as it should be in \eqref{amp-app}.
\section{ Evolution coming from dissipative re-normalization and approximation}\label{V}

\subsection{Re-normalized S-propagator}\label{S-propagator}

The re-normalization method introduced in Sec. \ref{renormalization} enables us to greatly simplify the formula for the the S-propagator evolution of the atom-field. As will be seen below after the re-normalization the evolution of the S-matrix, in our approximation regime, can be surprisingly fully determined by the normal ordered terms, the Lamb Shift and the decay rate. In fact, the Lamb Shift and the decay rate account for the contribution of all non-normal ordered terms in Dyson series. Using Eqs. \eqref{H0F}-\eqref{Frt}  the time-evolution operator, in the interaction picture with respect to $H_{0_r}$, takes the form:
\begin{eqnarray}\label{Uint}\nonumber
U_{I_r}(t,0)&=&\mathds{1}-i\int^{t}_{0}dt_1\ \tilde{H}_{I_r}(t_1)\\\nonumber
&+&(-i)^2\int^{t}_{0}dt_1\int^{t_1}_{0}dt_2\tilde{H}_{I_r}(t_1)\tilde{H}_{I_r}(t_2)+\dots\\
\end{eqnarray}
Applying now the approximation given in Eqs. \eqref{+approximation} and \eqref{-approximation} the evolution of S-propagator elements, i.e. $S_{I_r,ij}(t,0)\equiv\langle i|U_{I_r}(t,0)|j\rangle$, will take the simple form given by our main theorem.
\begin{thm}\label{thm1}
The evolution of S-propagator elements $S_{I_r,ij}(t,0)\equiv\langle i|U_{I_r}(t,0)|j\rangle$ of the whole system
is given by:
\begin{widetext}
\begin{equation}\label{th1}
S_{I_r,11}(t,0)\simeq \mathbb{I}_F+\sum_{n=1}(-i)^{2n}\int_{0}^{t}dt_1\ldots\int_{0}^{t_{2n-1}}dt_{2n}\ 
A'(t_2)\ldots A'(t_{2n})A(t_1)\ldots A(t_{2n-1}),
\end{equation}
\begin{equation}\label{th2}
S_{I_r,00}(t,0)\simeq\mathbb{I}_F+\sum_{n=1}(-i)^{2n}\int_{0}^{t}dt_1\ldots\int_{0}^{t_{2n-1}}dt_{2n}\ 
A'(t_1)\ldots A'(t_{2n-1})A(t_2)\ldots A(t_{2n}),
\end{equation}
\begin{equation}\label{th3}
S_{I_r,01}(t,0)\simeq\sum_{n=1}(-i)^{2n-1}\int_{0}^{t}dt_1\ldots\int_{0}^{t_{2n-2}}dt_{2n-1}\ 
A'(t_1)\ldots A'(t_{2n-1})A(t_2)\ldots A(t_{2n-2}),
\end{equation}
\begin{equation}\label{th4}
S_{I_r,10}(t,0)\simeq\sum_{n=1}(-i)^{2n-1}\int_{0}^{t}dt_1\ldots\int_{0}^{t_{2n-2}}dt_{2n-1}\ 
A'(t_2)\ldots A'(t_{2n-2})A(t_1)\ldots A(t_{2n-1}),
\end{equation}
\end{widetext}
where $\mathbb{I}_F$ is the identity matrix in the field space.
\end{thm}
Here we present a sketch for the proof (see Appendix \eqref{proof-thm1} for the whole proof).

\textit{Sketch of the proof:} Here we illustrate how the contribution from non-normally ordered terms, from different orders in Dyson series, cancel each other out. For the first order, in Dyson series with the initial exited state of the atom $|1\rangle$, the surviving term is
\begin{equation}\label{246}
\langle 1|\tilde{H}_{I_r}^1|1 \rangle\longrightarrow ib,
\end{equation}
where $\tilde{H}_{I_r}$ was defined in Eq. (\ref{eq:HI_A}) and the superscript 1 indicates the order in Dyson series. Now integrating over time we get
\begin{equation}\label{247}
(-i)\langle1|\int_{0}^{t}dt_1\tilde{H}^1_{I_r}|1\rangle=\cbl{\underbrace{bt}_{(I)}}\blk.
\end{equation}
For the second order, the surviving terms are
\begin{equation}\label{248}
\langle1|\tilde{H}_{I_r}^1\tilde{H}_{I_r}^2|1\rangle\longrightarrow a_1a^\dagger_2+(ib)^2=a^\dagger_2a_1+\delta_{12}+(ib)^2.
\end{equation}
Now integrating over time we get
\begin{align}\label{249}\nonumber
&\langle1|\int_{0}^{t}dt_1\int_{0}^{t_1}dt_2\tilde{H}_{I_r}^1\tilde{H}_{I_r}^2|1\rangle = \int_{0}^{t}dt_1\int_{0}^{t_1}dt_2\Big\{A'(t_2)A(t_1)\\ 
&+ F^r(t_1-t_2)+(ib)^2\Big\},
\end{align}
where $F(t_1-t_2)$ comes from the Dirac delta $\delta_{12}:=\delta(\omega_1-\omega_2)$ on the right hand of Eq. (\ref{248}). Now exploiting Eq. \eqref{+approximation} and the discussion in Subsection \ref{subC} we have the following approximation under the time integral for all times
\begin{equation}\label{251}
F^r(t_i-t_j)\approx b\delta(t_i-t_j).
\end{equation}
Therefore, we have the following substitution
\begin{equation}\label{251a}
\delta(\omega_i-\omega_j)\ \quad \Rightarrow\ \quad F^r(t_i-t_j)\ \quad \Rightarrow\ \quad b\delta(t_i-t_j).
\end{equation}
It should be noted that Eq. \eqref{251a} means that integrating $F^r(t_i-t_j)$ over time should give the following equality
\begin{equation}
    \int_0^t dt_i \int_0^{t_i} dt_j F^r(t_i-t_j) \approx \int_0^t dt_i \int_0^{t_i} dt_j\ b\delta(t_i-t_j) = bt.
\end{equation}
In Appendix \ref{DeltaPlot} we actually show that for an Ohmic distribution of $f(\omega)$ with cutoff on the frequency one can safely approximate the integral by $bt$. \blk Then Eq. (\ref{249}) reads
\begin{widetext}
\begin{eqnarray}\label{252}\nonumber
(-i)^2\langle1|\int_{0}^{t}dt_1\int_{0}^{t_1}dt_2\tilde{H}_{I_r}^1\tilde{H}_{I_r}^2|1\rangle&=& (-i)^2
\int_{0}^{t}dt_1\int_{0}^{t_1}dt_2\ \{A'(t_2)A(t_1)+b\delta(t_1-t_2)+(ib)^2\}\\
&=& (-i)^2\int_{0}^{t}dt_1\int_{0}^{t_1}dt_2\ A'(t_2)A(t_1)
- \cbl{\underbrace{bt}_{(I)}}\blk + \cred{\underbrace{b^2\int_{0}^{t}dt_1\int_{0}^{t_1}dt_2}_{(II)}}\blk.
\end{eqnarray}
\end{widetext}
For the third order, the surviving terms are
\begin{align}\label{253}\nonumber
\langle1|\tilde{H}_{I_r}^1\tilde{H}_{I_r}^2\tilde{H}_{I_r}^3|1\rangle &\longrightarrow ib(a_1a^\dagger_2+a_2a^\dagger_3)+(ib)^3 \\ 
&=ib(a^\dagger_2a_1+a^\dagger_3a_2+\delta_{12}+\delta_{23})+(ib)^3.
\end{align}
The contribution of $a_1a^\dagger_3$ is zero. Generally the terms $(i\delta\omega+\gamma)a_ia^\dagger_j$ with $j>i+1$ (like $a_1a^\dagger_3$) vanish. For example
\begin{widetext}
\begin{equation}\label{209b}
\sigma^+a_1(i\delta\omega+\gamma)|1\rangle\langle1|\sigma^-a^\dagger_3=(i\delta\omega+\gamma)a_1\underbrace{|1\rangle\langle 0|}_{\sigma^+}\underbrace{|1\rangle|v\rangle\langle v|\langle 1|}_{|1\rangle\langle1|\otimes|v\rangle\langle v|}\underbrace{|0\rangle\langle 1|}_{\sigma^-}a^\dagger_3=0,
\end{equation}
where $|v\rangle$ is the vacuum state of the field. Now integrating over time we get
\begin{eqnarray}\label{254}\nonumber
(-i)^3\langle1|\int_{0}^{t}dt_1\int_{0}^{t_1}dt_2\int_{0}^{t_2}dt_3\ \tilde{H}_{I_r}^1\tilde{H}_{I_r}^2\tilde{H}_{I_r}^3|1\rangle&=&(-i)^3
\int_{0}^{t}dt_1\int_{0}^{t_1}dt_2\int_{0}^{t_2}dt_3\Big\{ib[A'(t_2)A(t_1)+A'(t_3)A(t_2)\\
&+&F^r(t_1-t_2)+F^r(t_2-t_3)]+(ib)^3\Big\}.
\end{eqnarray}
Therefore, using Eq. (\ref{251}) we get
\begin{eqnarray}\label{255}\nonumber
(-i)^3\langle 1|\int_{0}^{t}dt_1\int_{0}^{t_1}dt_2\int_{0}^{t_2}dt_3\ \tilde{H}_{I_r}^1\tilde{H}_{I_r}^2\tilde{H}_{I_r}^3|1\rangle&=&(-i)^3
\int_{0}^{t}dt_1\int_{0}^{t_1}dt_2\int_{0}^{t_2}dt_3\ \Big\{ib[A'(t_2)A(t_1)+A'(t_3)A(t_2)\\\nonumber
&+&b\delta(t_1-t_2)+b\delta(t_2-t_3)]+(ib)^3\Big\}\\\nonumber
&=&\Big\{-\cbr{\underbrace{b
\int_{0}^{t}dt_1\int_{0}^{t_1}dt_2\int_{0}^{t_2}dt_3\ [A'(t_2)A(t_1)+A'(t_3)A(t_2)]}_{(III)}}\blk\\
&-&\cred{\underbrace{2b^2\int_{0}^{t}dt_1\int_{0}^{t_1}dt_2}_{(II)}}\blk
+\cgr{\underbrace{b^3\int_{0}^{t}dt_1\int_{0}^{t_1}dt_2\int_{0}^{t_2}dt_3}_{(IV)}}\blk\Big\},
\end{eqnarray}
\end{widetext}
As wee can see, the terms with the same number (or colour) already cancel out or they will with terms from higher orders. \qedsymbol{}

In Schr\"{o}dinger picture, the evolution of S-propagator elements $S_{lm}(t,0)$ read
\begin{equation}\label{Sth}
S_{lm}(t,0)=e^{-l(i\omega_A+\gamma) t}e^{-iH_Ft}S_{I_r,lm}(t,0),\ \quad l,m=0,1
\end{equation}
where $H_F$ is the free Hamiltonian of the field given in Eq. \eqref{H0r}. Using coherent states of the form $|\{\alpha\}\rangle \equiv|\alpha_1, \alpha_2, \ldots , \alpha_n \rangle$ and Theorem \ref{thm1} the matrix elements $S^{\beta\alpha}_{I_r,ij}(t,0)\equiv\langle\{\beta\},i|S_{I_r}(t,0)|j,\{\alpha\}\rangle$ can also be readily obtained as
\begin{widetext}
\begin{equation}\label{Nt1}
S^{\beta\alpha}_{I_r,11}(t,0)\simeq\langle\{\beta\}|\{\alpha\}\rangle\left(1+\sum_{n=1}(-i)^{2n}\int_{0}^{t}dt_1\ldots\int_{0}^{t_{2n-1}}dt_{2n}\ 
A'_\beta(t_2)A_\alpha(t_1)\ldots A'_\beta(t_{2n})A_\alpha(t_{2n-1})\right).
\end{equation}
\begin{equation}\label{Nt2}
S^{\beta\alpha}_{I_r,00}(t,0)\simeq\langle\{\beta\}|\{\alpha\}\rangle\left(1+\sum_{n=1}(-i)^{2n}\int_{0}^{t}dt_1\ldots\int_{0}^{t_{2n-1}}dt_{2n}\ 
A'_\beta(t_1)A_\alpha(t_2)\ldots A'_\beta(t_{2n-1})A_\alpha(t_{2n})\right),
\end{equation}
\begin{equation}\label{Nt3}
S^{\beta\alpha}_{I_r,01}(t,0)\simeq\langle\{\beta\}|\{\alpha\}\rangle\left(\sum_{n=1}(-i)^{2n-1}\int_{0}^{t}dt_1\ldots\int_{0}^{t_{2n-2}}dt_{2n-1}\ 
A'_\beta(t_1)A_\alpha(t_2)\ldots A'_\beta(t_{2n-1})A_\alpha(t_{2n-2})\right),
\end{equation}
\begin{equation}\label{Nt4}
S^{\beta\alpha}_{I_r,10}(t,0)\simeq\langle\{\beta\}|\{\alpha\}\rangle\left(\sum_{n=1}(-i)^{2n-1}\int_{0}^{t}dt_1\ldots\int_{0}^{t_{2n-2}}dt_{2n-1}\ 
A'_\beta(t_2)A_\alpha(t_1)\ldots A'_\beta(t_{2n-2})A_\alpha(t_{2n-1})\right),
\end{equation}
\end{widetext}
where
\begin{equation}\label{th5}
A_\alpha(t)\equiv e^{i\Omega t}f_{\alpha}(t),\ \quad A'_\beta(t)\equiv e^{-i\Omega t}f_{\beta}^*(t),
\end{equation}
with
\begin{equation}\label{FourierTransformF}
f_{\alpha}(t)=\int dk\ f(\omega_k)\alpha(\omega_k)e^{-i\omega_kt}.
\end{equation}
In Schr\"{o}dinger picture, the matrix elements $S^{\beta\alpha}_{lm}(t,0)\equiv\langle\{\beta\},l|S(t,0)|m,\{\alpha\}\rangle$, $l,m=0,1$ are obtained as
\begin{equation}\label{th6}
S^{\beta\alpha}_{lm}(t,0)=e^{-l\gamma t}e^{-i(l\omega_A+\Theta)t}\langle\{\beta\},l|S_{I_r}(t,0)|m,\{\alpha\}\rangle,
\end{equation}
where $\Theta=\int dk\ \omega_k\beta^*(\omega_k)\alpha(\omega_k)$.

\subsection{S-propagator differential equation}
The (non-unitary) re-normalized interaction picture evolution operator $U_{I}^r(t)$ satisfies the following differential equation
\begin{equation}\label{th7}
\frac{d}{dt}U_{I_r}(t) = -i \tilde{H}_{I_r}(t)U_{I_r}(t).
\end{equation}
The result obtained above is the derivation of the useful approximation for the map $U_{I_r}(t)$ given by a much simpler map $S_{I}(t)$ which can be treated as a certain weak coupling approximation leading to normal order expression in terms of  field operators. This new approximate dynamics satisfies the following differential equation (notice the unusual ordering of the operators!)
\begin{equation}\label{th8}
\frac{d}{dt}S_{I_r}(t) =  -i \left(\sigma^+ S_{I_r}(t)A(t) +\sigma^- A'(t) S_{I_r}(t) \right)  ,\quad S_{I}(0) = 1.
\end{equation}
The explicit form of $S_{I}(t)$ can be obtained using the Dyson series expansion for Eq. \eqref{th8} and the properties of $\sigma^{\pm}$. It corresponds to Eqs. \eqref{th1}-\eqref{th4}. In Schr\"{o}dinger picture we have
\begin{align*}\label{Sd1}
\frac{d}{dt}S(t)&=\frac{d}{dt}(e^{-iH_{0_r}t}S_{I_r})\\\nonumber
&=-iH_{0_r}e^{-iH_{0_r}t}S_{I_r}+e^{-iH_{0_r}t}\frac{d}{dt}S_{I_r}(t)\\\nonumber
&=-iH_{0_r}S(t)-ie^{-iH_{0_r}t}\Big(\sigma^+ e^{iH_{0_r}t}S(t)A(t) \\\nonumber
&+\sigma^- A'(t) e^{iH_{0_r}t}S(t)\Big)\\\nonumber
&=-iH_{0_r}S(t)-i\Big(\sigma^+S(t)\int dk\ f(\omega_k)e^{-i\omega_kt}a_{\omega_k} \\\nonumber
&+\sigma^- \int dk\ f^*(\omega_k)a^\dagger_{\omega_k}S(t)\Big).\numberthis
\end{align*}
A useful representation of $S_{I_r}(t)$ can be written in terms of the partial matrix element
\begin{equation}
S_{I_r}^{\beta\alpha} (t)  \equiv \langle\{\beta\}| S_{I_r}(t) |\{\alpha\}\rangle
\label{Scoh}
\end{equation}
with respect to the coherent states $|\{\alpha\}\rangle ,|\{\beta\}\rangle$ defined as before. Each of them satisfies the following evolution equation for the $2\times 2$ matrix  $S_I^{\beta\alpha} (t)  =  [S^{\beta\alpha}_{I,ij} (t)] ,\,  i,j = 0,1$;
\begin{equation}\label{eqScoh}
\frac{d}{dt}S_{I_r}^{\beta\alpha}(t) = -i\Bigl(A_{\alpha}(t)\sigma^+ + A'_{\beta}(t)\sigma^-\Bigr)S_{I_r}^{\beta\alpha} (t),
\end{equation}
and the initial value $S_{I_r}^{\beta\alpha}(0) = \langle\{\beta\}|\{\alpha\}\rangle$.

\subsection{Particular examples}

1) The simplest object is the survival amplitude of the atomic excited state in the vacuum field defined as $\langle\{0\},1|e^{-iHt}|1,\{0\}\rangle$. Then
\begin{eqnarray}\label{WWapprox}\nonumber
\langle \{0\},1| e^{-iHt} |1,\{0\}\rangle &=&  e^{(-i\omega_A - \gamma)t} \langle\{0\},1|U_r(t)|1,\{0\}\rangle\\ \nonumber
&=& e^{(-i\omega_A - \gamma)t} S^{00}_{11}\\
&\simeq& e^{(-i\omega_A - \gamma)t},
\end{eqnarray}
when the last line comes from our Theorem \ref{thm1} in which we assumed that our approximation (Eq. \eqref{approximation}) holds for long time. Eq. \eqref{WWapprox} reproduces the Wigner-Weisskopff result.

2) Under the evolution for initial state $|1,\{0\}\rangle$ the norm of the state at time $t$ is approximately preserved (see Appendix \ref{PE62})
\begin{equation}\label{Nm62}
\langle\psi(t)|\psi(t)\rangle=\langle \{0\},1|S^\dagger(t,0)S(t,0)|1,\{0\}\rangle\approx1.
\end{equation}
3) Taking  $\beta_k = \alpha_k$  where $\alpha$ represent the initial state of the field we can compute the final state of the atom  $|\psi(t)\rangle$  in the improved semi-classical  approximation which assumes that the total state remains a product state of the atom and the freely evolving field $|\alpha_t\rangle$
\begin{align*}\label{semclass1}
|\psi(t)\rangle &=  e^{t (-i\omega_A - \gamma)|1\rangle\langle 1|} S_{I_r}^{\alpha\alpha}(t) |\psi(0)\rangle\\\nonumber
&=U_\alpha(t)|\psi(0)\rangle,\numberthis
\end{align*}
where
\begin{equation}\label{semclass2}
U_\alpha(t) =  e^{t (-i\omega_A - \gamma)|1\rangle\langle 1|} S_{I_r}^{\alpha\alpha}(t).
\end{equation}
The differential equations corresponding to Eq. \eqref{semclass1} take the form
\begin{equation}\label{diffsemclass1}
\frac{d}{dt}|\psi(t)\rangle =  -i\bigl[(\omega_A -i \gamma)|1\rangle\langle 1| +  f_{\alpha}(t)\sigma^+     +   f^*_{\alpha}(t)   \sigma^- \bigr] |\psi(t)\rangle,
\end{equation}
and 
\begin{equation}\label{diffsemclass2}
\frac{d}{dt}U_\alpha(t)=  -i\bigl[(\omega_A -i \gamma)|1\rangle\langle 1| +  f_{\alpha}(t)\sigma^+     +   f^*_{\alpha}(t)   \sigma^- \bigr] U_\alpha(t).
\end{equation}
4) The most general situation is described by the initial field state written in the Glauber P-representation
\begin{equation}
\rho^{field}(0) =  \int \mathcal{D}\alpha \, P_0(\alpha) |\alpha\rangle\langle\alpha|
\label{P-rep}
\end{equation}
where we use symbolic notation for functional integral properly defined by a limit procedure. Here, the functional $P(\alpha)$ takes values in $2\times 2$ matrices.
Then we can compute the ($2\times 2$ matrix valued) Husimi Q-function for the final state , $Q_t(\beta)\equiv  \langle\{\beta\}|\rho(t)|\{\beta\}\rangle$ as 
\begin{widetext}
\begin{equation}
Q_t(\beta) =   e^{-iH^r_0 t} \left[\int \mathcal{D}\alpha \, P_0(\alpha) S_I^{\beta\alpha}(t) |\psi(0)\rangle\langle\psi(0)|\left(S_I^{\alpha\beta}(t)\right)^{\dagger}\right] \left(e^{-iH^r_0 t}\right)^{\dagger}.
\label{P-rep}
\end{equation}
\end{widetext}
The above expression is, in principle, computable for numerous examples of initial states.

\section{Illustration}\label{VI}
In this section we would like to apply our approximation to different settings and obtain the dynamics for the atom. First, we apply it to the case in which the initial state of the field consists of a continuous superposition of one photon in different frequencies. 
Later, we treat the evolution for the coherent state of the field and get the Optical Bloch equations, which shows that our approach reproduces 
known results. 
\begin{figure}[b]
    \centering
    \includegraphics[scale=0.40]{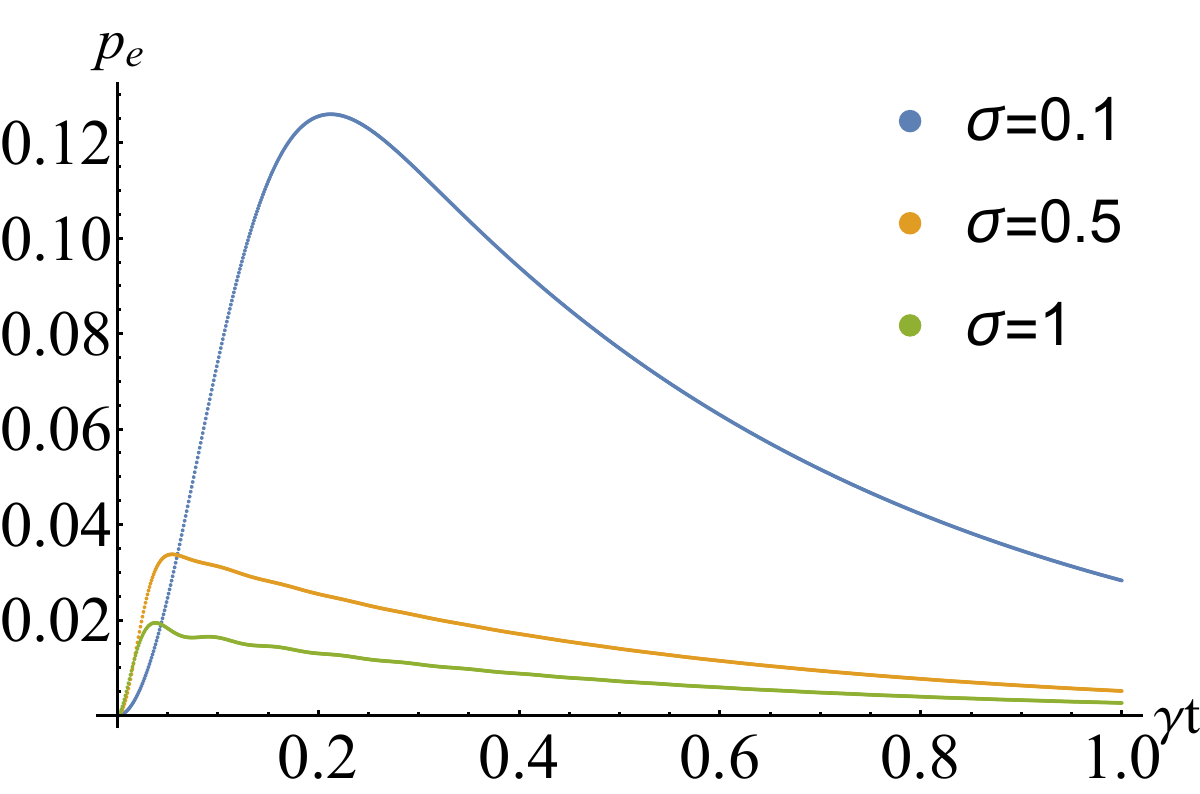}\\
    
    \includegraphics[scale=0.40]{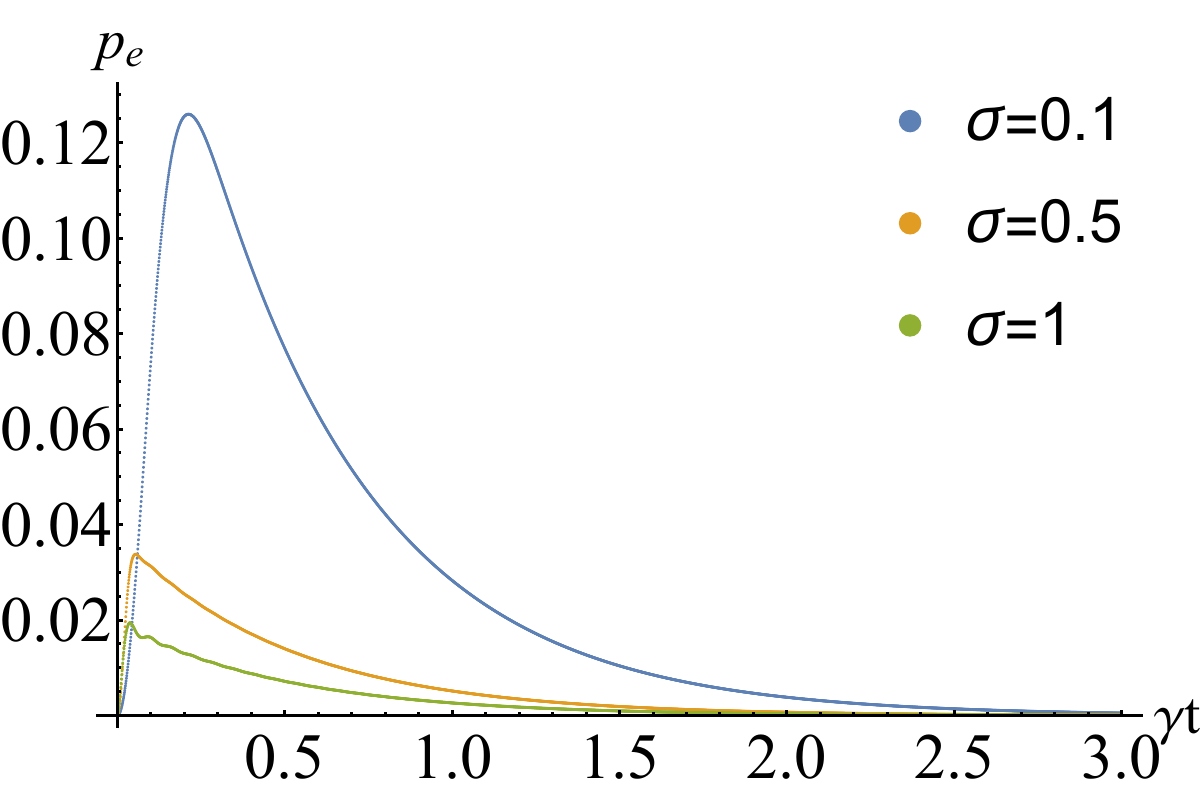}\\
    \caption{Numerical plots of the population of the excited state of the atom as a function of $\gamma t$ for $\omega_A=1$, $\gamma=0.01$, and $\kappa = \sqrt{\frac{\gamma}{\pi \omega_A}}$ for different $\sigma$ and for short (top) and long (bottom) time scales. 
    }
    \label{excited-population-short}
    \label{excited-population}
\end{figure}
\subsection{One photon }
In this subsection, we consider that the initial state of the atom and field is given by
\begin{equation}
    \rho(0)=|0\rangle \langle 0| \otimes |h\rangle \langle h|,
\end{equation}
where $|0\rangle$ corresponds to the ground state of the atom and $|h\rangle = \int dk h(\omega_k)|\{1_k\}\rangle$ and $h(\omega_k)=\frac{1}{\pi^{1/4}\sqrt{\sigma}}e^{-(\omega_k-\omega_A)^2/(2\sigma^2)}$, which is a Gaussian function. One could extend our calculations for an arbitrary state of the atom but because it is not very illuminating, we restrict ourselves to this simpler case. As our state is pure, we calculate the following (in the interaction picture with respect to $H_{0_r}$)
\begin{align*}\nonumber
    &|\psi(t)\rangle=S_{I_r}(t)|0,h \rangle
    = S_{I_r,00}(t) |0, h \rangle + S_{I_r,10}(t) |1,h\rangle\\ 
   &= |0, h \rangle - \int_0^t dt_1 \int_0^{t_1} dt_2\ A'(t_1)A(t_2) |0,h\rangle 
    - i \int_0^t dt_1 A(t_1)|1,h\rangle\\\nonumber &=|0, h\rangle \\
    &- \int_0^t dt_1 \int_0^{t_1} dt_2\ \int dk'\ \int dk\ f'(\omega_{k'},t_1)f(\omega_{k},t_2) h(\omega_k) |0,\{1_{k'}\}\rangle \\
    &- i \int_0^t dt_1 \int dk\ f(\omega_{k},t_1) h(\omega_k) |1,\{0\}\rangle.\numberthis
\end{align*}
The state of the atom at time $t$ is given by $\tilde{\rho}_A(t) = \operatorname{Tr}_F \left(|\psi(t)\rangle \langle \psi(t)|\right)$. As our dynamics is trace-preserving (in the Schr\"odinger picture), we are only interested in the population of the excited level of the atom. Then, in the Schr\"odinger picture we have
\begin{align}\nonumber
    &p_e(t)=e^{-2\gamma t}\left|\int_0^t dt_1 \int dk\ f(\omega_k,t_1)h(\omega_k)\right|^2\\
    &=e^{-2\gamma t}\frac{1}{\sqrt{\pi}\sigma}\left|\int dk\ f(\omega_k)\frac{e^{\gamma t}e^{-i(\omega_k - \omega_A)t}-1}{\gamma -i(\omega_k - \omega_A)}e^{-\left(\omega_k-\omega_A \right)^2/(2\sigma^2)}\right|^2,
\end{align}
where $f(\omega_k)=\kappa \sqrt{\omega_k}$. In Fig. \eqref{excited-population} we can see the numerical solution for $p_e(t)$. It is clear that the population decays exponentially as $t\rightarrow \infty$ as ensured by our approximation. However, by decreasing $\sigma$, one can see that more time is needed for the atom to decay. The reason is that the state of the field approaches a sharp state of one photon in one specific mode, which will be de-localized in time and the atom will be interacting with the photon at all times. Because of this interaction, the atom will absorb the photon with almost constant probability for all values of time, making it impossible for the atom to decay. 

\blk

\subsection{Coherent state}
In order to perform the calculations of this subsection, we first need to write the evolution for the propagator. Making use of Eq. \eqref{eqScoh} in Schr\"odinger picture we have
\begin{equation}\label{diffsemclass0}
U^{\beta\alpha}(t)\equiv\frac{1}{\mathcal{N}_{\beta \alpha}}e^{t(-i\omega_A - \gamma)|1\rangle\langle 1|}S_{I_r}^{\beta\alpha}(t),
\end{equation}
where $\mathcal{N}_{\beta\alpha}\equiv\langle\{\beta\}|\{\alpha\}\rangle$. The now normalized Schr\"{o}dinger picture propagator ($U^{\beta\alpha}(0)=\mathbb{I}$) satisfies the following evolution equation
\begin{align*}\label{ne}
\frac{d}{dt}U^{\beta\alpha}(t)&= -i\Big[(\omega_A - i\gamma)|1\rangle\langle 1| + f_{\alpha}(t)\sigma^+ + f^*_{\beta}(t)\sigma^-\Big]\\\nonumber
&\times U^{\beta\alpha}(t),\numberthis
\end{align*}
where $f_{\alpha}(t)$ is defined in Eq. \eqref{FourierTransformF}. It is more convenient to define $\xi=\beta-\alpha$ and change the propagator parametrization to $U^{\alpha \beta}(t) \equiv U^{\xi}_{\alpha}(t)$. In this way, the evolution is given by
\begin{align*}\label{diffsemclass1}
\frac{d}{dt}U^{\xi}_{\alpha}(t)&=-i\Big[(\omega_A -i \gamma)|1\rangle\langle 1| +  f_{\alpha}(t)\sigma^+ + f^*_{\alpha}(t)\sigma^-\\\nonumber
&+f^*_{\xi}(t)\sigma^-\Big] U^{\xi}_{\alpha}(t)\\\nonumber
&=-i\Big[H_{A_{r}}+H_\alpha(t)+ f^*_{\xi}(t)\sigma^-\Big]U^{\xi}_{\alpha}(t).\numberthis
\end{align*}
where 
\begin{align}
&H_{A_{r}}\equiv(\omega_A-i\gamma)|1\rangle\langle 1|,\label{H-ar}
\end{align}
and
\begin{equation}\label{diffsemclass4}
H_\alpha(t)\equiv f_{\alpha}(t)\sigma^+ + f^*_{\alpha}(t)\sigma^-.
\end{equation}
We now move to the (non-unitary) interaction picture with respect to the Hamiltonian $H_{A_r}+H_\alpha(t)$. Therefore, the propagator on this picture will be 
\begin{equation}
   \tilde{U}^{\xi}_{\alpha}=\mathcal{T}e^{i\int_0^t dt'(H_{A_r}+H_\alpha(t'))}\mathcal{T}e^{-i\int_0^t dt'(H_{A_r}+H_\alpha(t')+f^*_{\xi}(t')\sigma^-)}.
\end{equation}
Hence we can write (see Appendix \ref{proof-sigma} for the complete derivation)
\begin{align}\label{evolution-propagator}
\frac{d}{dt}\tilde{U}^{\xi}_{\alpha}(t)&=-i f^*_{\xi}(t)\tilde{\sigma}^-(t)\tilde{U}^{\xi}_{\alpha}(t),
\end{align}
in which
\begin{align*}\label{diffsemclass4}
\tilde{\sigma}^-(t) &= \mathcal{T}e^{i\int_0^t dt'(H_{A_r}+H_\alpha(t'))t'}\sigma^-\mathcal{T}e^{-i\int_0^t dt'(H_{A_r}+H_\alpha(t'))t'}.\numberthis
\end{align*}
Before further calculations, we make the following change $\gamma\rightarrow\frac{\gamma}{2}$ (so $b+b^*=\gamma$ now) for ease of notation. We rewrite Eq. \eqref{diffsemclass1} as
\begin{align*}\label{nv1}
\frac{d}{dt}U^{\xi}_{\alpha}(t)&=-i\Big[\omega_AP_1 +  f_{\alpha}(t)\sigma^+ + f^*_{\alpha}(t)\sigma^-\big] U^{\xi}_{\alpha}(t)\\\nonumber
&-\Big[\frac{\gamma}{2}P_1+if^*_{\xi}(t)\sigma^-\Big] U^{\xi}_{\alpha}(t),\numberthis
\end{align*}
where $P_1\equiv|1\rangle\langle1|$. Defining $U_{\alpha}(t,0)$ as
\begin{equation}\label{nv2}
     U_{\alpha}(t,0)=\mathcal{T}e^{-i\int_0^t dt'H_{\alpha}(t')},
\end{equation}
where $H_{\alpha}(t)\equiv\omega_AP_1 +  f_{\alpha}(t)\sigma^+ + f^*_{\alpha}(t)\sigma^-$, we can expand $U^{\xi}_{\alpha}(t)$ in Dyson series as 
\begin{align*}\label{nv3}
     U^{\xi}_{\alpha}(t)&= U_{\alpha}(t,0)\Big[\mathbb{I}-\int_0^t ds_1\left(\frac{\gamma}{2}\tilde{P}_1(s_1)+if^*_{\xi}(s_1)\tilde{\sigma}^-(s_1)\right)\\\nonumber
     &+ \int_0^t ds_1\int_0^{s_1} ds_2\left(\frac{\gamma}{2}\tilde{P}_1(s_1)+if^*_{\xi}(s_1)\tilde{\sigma}^-(s_1)\right)\\\nonumber
     &\times \left(\frac{\gamma}{2}\tilde{P}_1(s_2)+if^*_{\xi}(s_2)\tilde{\sigma}^-(s_2)\right)\Big]\\\nonumber
     &+ \mathcal{O}(\gamma^3),\numberthis
\end{align*}
where $\tilde{X}\equiv U_{\alpha}^{\dagger}(s_1,0) X U_{\alpha}(s_1,0)$. Then the reduced density matrix of the atom at time $t$, up to the first order in $\gamma$, reads
\begin{align*}\label{nv4}
\rho_A(t)&= \Lambda (t)\rho_A(0)\\
&\approx U_{\alpha}(t,0)\Big[\rho_A(0) - \frac{\gamma}{2}\int_0^t ds_1\Big(\tilde{P}_1(s_1)\rho_A(0)\\
&+\rho_A(0)\tilde{P}_1(s_1)\Big) + \gamma\int_0^t ds_1\tilde{\sigma}^-(s_1)\rho_A(0)\tilde{\sigma}^+(s_1)\Big]\\
&\times U^\dagger_{\alpha}(t,0),\numberthis
\end{align*}
where we have applied the the usual approximation for $F^r$ (see Appendix \eqref{2ndorder} for all the calculations). Note that for the approximation, used above, to be valid it was assumed that we have a slow driving laser field with $f_\alpha(t)$ small enough such that $\tilde{\sigma}^-(t)$ is an enough slow varying function in the time scale of $1/\gamma$. 
Therefore the map $\Lambda(t)$, up to the first order in $\gamma$, may be written in the form
\begin{align*}\label{nv5}
     \Lambda (t)(\cdot)&= \mathcal{U}_{\alpha}(t,0) \Big[\mathbb{I}-\frac{\gamma}{2}\int_0^t ds_1\{\tilde{P}_1(s_1), (\cdot)\}\\\nonumber
     & + \gamma\int_0^t ds_1\  \tilde{\sigma}^-(s_1)(\cdot)\tilde{\sigma}^+(s_1)\Big] + \mathcal{O}(\gamma^2),\numberthis
\end{align*}
where the super-operator $\mathcal{U}_{\alpha}(t,0)[\cdot]=U_{\alpha}(t,0)[\cdot]U^\dagger_{\alpha}(t,0)$. Hence, inspired by Eq. \eqref{nv5} we conjecture that the following differential equation for the evolution of $\Lambda(t)$, in the interaction picture with respect to $H_{\alpha}(t)$ (defined at the beginning of the subsection), holds
\begin{align*}\label{nv51}
     \frac{d\tilde{\Lambda} (t)(\cdot)}{dt}&\approx -\frac{\gamma}{2}\{\tilde{P}_1(t), (\cdot)\} + \gamma  \tilde{\sigma}^-(t)(\cdot)\tilde{\sigma}^+(t).\numberthis
\end{align*}
Our conjecture is indeed true -at least, up to the second order (see the proof in the Appendix \eqref{2ndorder})- and we strongly think it will be preserved for all the orders. Finally, we would like to mention that in the simpler case in which the state starts in the vacuum state ($\alpha=0$) one obtains that the evolution of the reduced state of the atom is given by the Gorini-Kossakowski-Lindblad-Sudarshan (GKLS) master equation \cite{gorini1976completely,lindblad1976generators}.\blk



\section{Conclusion}\label{VII}


In this work, we investigated the interaction of a two-level atom with a quantized continuous-mode laser field. Using a re-normalized method we managed to write the evolution of the atom-field system in a form that depends only on normally ordered creation and annihilation operators of the field, the decay rate of the atom, and its natural frequency $\omega_A$. In fact, the decay rate and the Lamb Shift account for the contribution of all remaining terms, which appear after the normal ordering of the creation and annihilation operators of the field in the Dyson series. Furthermore, we showed that our approach reproduces the previously known optical Bloch equations. 
We have studied just a two-level atom, but the need for a generalization to a d-level system may arise naturally. This would be useful for studying typical setups of V-systems or $\Lambda$-systems.

Aside from the Quantum Optics realm, we reckon it may be beneficial to translate the results presented here to the framework of Feynman diagrams to understand better the evolution of the system at all times. In addition, the re-normalization method introduced here may have some potential in dealing with cut-off terms that arise in the evolution as well as giving some flavour in dealing with divergences that appear in Open Quantum systems in general.  

Last but not least, will be the investigation of quantumness traces of gravitational field when interacting with gravitational wave detectors. It is of great interest to know what the quantumness effects of the gravitational wave are when interacting with the lengths of the arms of gravitational wave detectors. We hope that our proposed study can prove useful also in the above context.


\begin{acknowledgements}
R.R.R. acknowledges helpful discussions with Konrad Schlichtholz. We acknowledge support from the Foundation for Polish Science through IRAP project co-financed by EU within the
Smart Growth Operational Programme (contract no.2018/MAB/5).
M.H. is also supported by National Science Center, Poland, through grant OPUS (2021/41/B/ST2/03207). 
\end{acknowledgements}

\bibliographystyle{ieeetr}
\bibliography{references}
\onecolumngrid
\newpage

\appendix

\section{Proof of Eq. \eqref{-approximation}}\label{proof-approximation}

Employing Eqs. \eqref{Frt} and \eqref{+approximation} we have
\begin{equation}\nonumber
F^r(t)=e^{\gamma t}h(t)\approx b\delta(t),
\end{equation}
in which 
\begin{equation}\nonumber
h(t)\equiv\int dk\, e^{-i(\omega_k-\omega_A)t} |f(\omega_k)|^2.
\end{equation}
Then for $t<0$, denoting $u=-t$ we get
\begin{eqnarray}\nonumber
F^r(t)&=&F^r(-u)\\\nonumber
&=&e^{-\gamma u}h^*(u)\\\nonumber
&=&e^{-2\gamma u}\left(e^{\gamma u}h(u)\right)^*\\\nonumber
&\approx&e^{-2\gamma u}b^*\delta(u)\\\nonumber
&=&e^{2\gamma t}b^*\delta(-t)\\\nonumber
&=&e^{2\gamma t}b^*\delta(t)\\\nonumber
&=&b^*\delta(t),\\
\end{eqnarray}
where in the last equality we dropped the term $e^{2\gamma t}$ just because inside the time integrals this term is equal to 1 due to the presence of Dirac delta.

\section{Justification of the substitution in Eqs. \eqref{+approximationreplacement} and \eqref{-approximationreplacement}}\label{DeltaPlot}
\begin{figure}[b]
    \centering
    \includegraphics[scale=0.40]{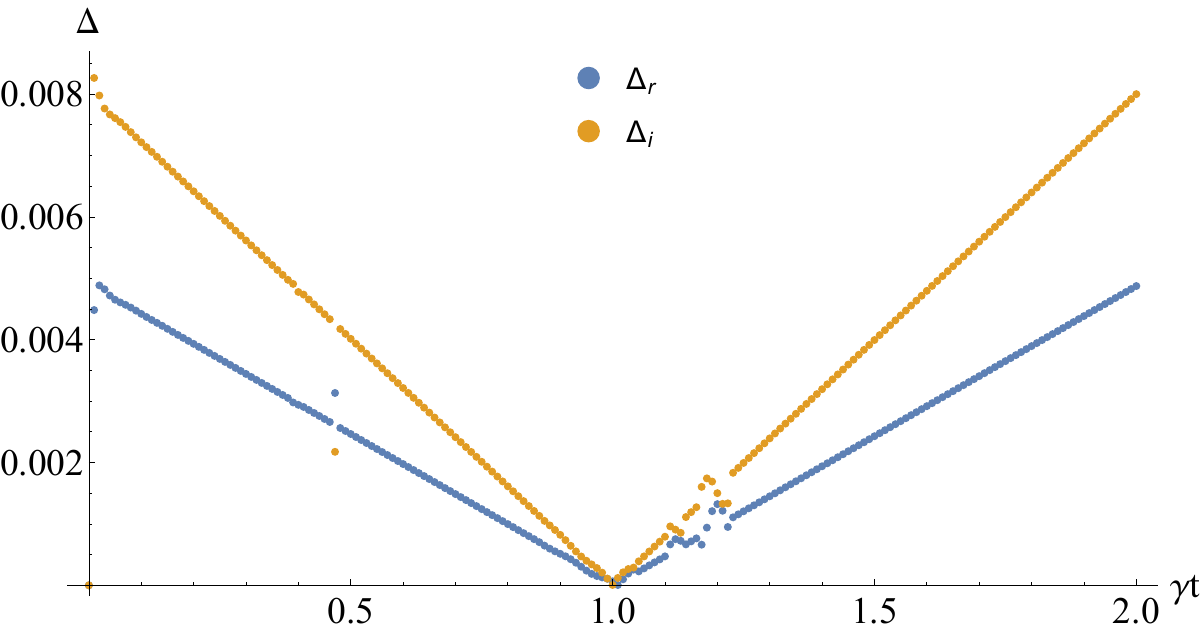}\\
    \caption{$\Delta(t)$ vs $\gamma t$ for $\gamma=0.01$, $\omega_A=1$ and $\omega_c=5$. Here $\Delta_r(t)\equiv\Re{\Delta(t)}-\gamma t$ and $\Delta_i(t)\equiv\Im{\Delta(t)}-\delta\omega t$.}
    \label{Delta}
\end{figure}
Below in Fig. \ref{Delta} we illustrate that for times up to $\gamma t = 2$, with the decay rate $\gamma=0.01$, our substitution introduced in Eqs. \eqref{+approximationreplacement} and \eqref{-approximationreplacement} is valid with an error of the order of $10^{-3}$ magnitude. Here $\Delta(t)$ is defined as
\begin{equation}
    \Delta(t) \equiv \int_0^t dt_i\int_0^{t_i} dt_j\ F^r(t_i-t_j) - bt,
\end{equation}
with $b=\gamma+i\delta\omega$ and the one-dimensional spectral density
\begin{equation}
    f(\omega)=\alpha \sqrt{\omega} e^{\frac{-|\omega|}{2\omega_C}},
\end{equation}
in which
\begin{equation}
    \alpha^2\equiv\frac{\gamma}{\pi \omega_A} e^{\frac{\omega_A}{\omega_c}}
\end{equation}
where $\omega_c$ is the cutoff frequency, $\omega_A$ the atomic frequency, and we have used the formula obtained for $\gamma$ given in Eq. \eqref{DP3}.

\section{Proof of Eq. \eqref{evolution-propagator}}\label{proof-sigma}

Assume that
\begin{equation}\label{diffsemclass12}
\frac{d}{dt}U(t) = \left(Z(t) + A(t)\right)U(t),
\end{equation}
and
\begin{equation}\label{diffsemclass13}
\frac{d}{dt}U_0(t) = Z(t)U_0(t).
\end{equation}
Using the fact that $U^{-1}_0(t)U_0(t)=\mathbb{I}$ we have
\begin{align*}\label{diffsemclass14}
\frac{d}{dt}\left(U^{-1}_0(t)U_0(t)\right) &= \frac{d}{dt}\left(U^{-1}_0(t)\right)U_0(t) + U^{-1}_0(t)\frac{d}{dt}U_0(t)\\\nonumber
&= \left (\frac{d}{dt}U^{-1}_0(t) + U^{-1}_0(t)Z(t)\right) U_0(t)\\\nonumber
&= 0,\numberthis
\end{align*}
which gives
\begin{equation}\label{diffsemclass15}
\frac{d}{dt}U^{-1}_0(t) = -U^{-1}_0(t)Z(t).
\end{equation}
Using the above equation we can write
\begin{align*}\label{diffsemclass16}
\frac{d}{dt}\tilde{U}(t) &= \frac{d}{dt}\big(U^{-1}_0(t)U(t)\big)\\\nonumber
&= \dfrac{d}{dt}U^{-1}_0(t)U(t) + U^{-1}_0(t)\frac{d}{dt}U(t)\\\nonumber
&= -U^{-1}_0(t)Z(t)U_0(t)U^{-1}_0(t)U(t)\\\nonumber
&+ U^{-1}_0(t)\big(Z(t) + A(t)\big)U_0(t)U^{-1}_0(t)U(t)\\\nonumber
&= U^{-1}_0(t)A(t)U_0(t)\tilde{U}(t)\\\nonumber
&= A'(t)\tilde{U}(t),\numberthis
\end{align*}
which completes the proof.
\section{Proof of Theorem \ref{thm1}}\label{proof-thm1}

Here the main result of the paper, i.e., Theorem \ref{thm1} will be proved. As a matter of fact, we shall only prove Eq. \eqref{th1}, and the three others are proved analogously. The proof of Eq. \eqref{th1} is given below in Proposition \ref{prop2} of this appendix. We shall need the following notation.
Let $I=\{(i_1,i_1+1),\ldots,(i_l,i_l+1)\}$,
be set of pairs of indices where
$i_1<i_2-1, i_2< i_3-1,\ldots, i_{l-1}\leq i_l-1$.
Then we will denote
\begin{align}
\label{eq:X-notation}
X_I(A A') =A(t_{i_1})A'(t_{i_1+1})\ldots
A(t_{i_{l}})A'(t_{i_l+1}),\quad
X_I(a a^\dagger) =a_{i_1}a^\dagger_{i_1+1}\ldots
a_{i_l}a^\dagger_{i_l+1} \nonumber \\
X_I(F)=
F(t_{i_1}-t_{i_1+1})\ldots
F(t_{i_{l}}-t_{i_l+1}),\quad
X_I(\delta)=
\delta(t_{i_1}-t_{i_1+1})\ldots
\delta(t_{i_{l}}-t_{i_l+1})
\end{align}
We note that normal ordering of creation annihilation operators translates into normal ordering of
$A,\tilde A$ operators, so that, in particular,  we have
\begin{equation}
:X_I\!: \,\,
=
 A'(t_{i_1+1})
A'(t_{i_1+2}) \ldots
A'(t_{i_l+1})
\,
A(t_{i_1})A(t_{i_{2}})\ldots
A(t_{i_l}),
\end{equation}
where $":x:"$ denotes normal ordering. For a set of pairs of neighboring natural numbers the following
\begin{equation}
I \sim [n]
\end{equation}
means that $I$ is set of pairs chosen from
the set $[n]\equiv\{1,\ldots,n\}$.
 (The set $I$ mey not be the set of all pairs).
For example $\{(1,2),(5,6)\}\sim [7]$.

Finally we shall denote:
\begin{equation}\label{N1}
\int dt^{n}\equiv\int_{0}^{t}dt_1\int_{0}^{t_1}dt_2\int_{0}^{t_2}dt_3\ldots\int_{0}^{t_{n-2}}dt_{n-1}\int_{0}^{t_{n-1}}dt_n.
\end{equation}

Having set the needed notation, we being with the following lemma.

\begin{lemma}\label{lemma1}
\begin{equation}\label{Lr}
\langle 1|\tilde{H}^{ 1} \tilde{H}^{ 2} \ldots \tilde{H}^{ n-1} \tilde{H}^{ n} |1\rangle=\left\{\begin{array}{cc}
{A(t_1)A'(t_2)\ldots
A(t_{n-1})A'(t_n)} & {\text {\textit{for n even,}}} \\
{0} & {\text {\textit{for n odd,}}}
\end{array}\right.
\end{equation}
\begin{equation}\label{L15}
\langle 0|\tilde{H}^{ 1} \tilde{H}^2\ldots \tilde{H}^{n-1}\tilde{H}^n|0\rangle=\left\{\begin{array}{cc}
{A'(t_1)A(t_2)\ldots
A'(t_{n-1})A(t_n)} & {\text {\textit{for n even,}}} \\
{0} & {\text {\textit{for n odd,}}}
\end{array}\right.
\end{equation}
\begin{equation}\label{L16}
\langle0|\tilde{H}^1\tilde{H}^2\ldots \tilde{H}^{n-1}\tilde{H}^n|1\rangle=\left\{\begin{array}{cc}
{A'(t_1)A(t_2)\ldots
A(t_{n-1})A'(t_n)} & {\text {\textit{for n odd,}}} \\
{0} & {\text {\textit{for n even,}}}
\end{array}\right.
\end{equation}
\begin{equation}\label{L17}
\langle1|\tilde{H}^1\tilde{H}^2\ldots \tilde{H}^{n-1}\tilde{H}^n|0\rangle=\left\{\begin{array}{cc}
{A(t_1)A'(t_2)\ldots
A'(t_{n-1})A(t_n)} & {\text {\textit{for n odd,}}} \\
{0} & {\text {\textit{for n even.}}}
\end{array}\right.
\end{equation}
\end{lemma}
\begin{proof}
Using Eq. \eqref{eq:HI_A} where $\tilde{H}_{I_r} \equiv \tilde{H}  $ we have
\begin{equation}\label{L11}
\tilde{H}^1\tilde{H}^2\ldots \tilde{H}^{n-1}\tilde{H}^n=
\sum_{s_1,s_2,\ldots,s_n=\pm}A^{(s_1)}(t_1) \ldots A^{(s_n)}(t_n)\sigma^{s_1}\ldots\sigma^{s_n},
\end{equation}
where
\begin{equation}\label{L12}
A^{(s_j)}(t_j)=
\left\{\begin{array}{ccl}
A(t_j) & \quad \text{for } s_j=+ &\\
A'(t_j) & \quad \text{for } s_j=-&.
\end{array}\right.
\end{equation}
Since $\sigma^{+}\sigma^{+}|1\rangle=0$ and $\sigma^{-}\sigma^{-}|1\rangle=0$ hence
\begin{equation}\label{L13}
\langle1|\sigma^{s_1}\ldots\sigma^{s_n}|1\rangle=\left\{\begin{array}{cc}
{1} & {\text {\textit{if} $s_1\ldots s_n=+-\ldots+-$,}} \\
{0} & {\text {\textit{otherwise.}}}
\end{array}\right.
\end{equation}
This means, that firstly $n$ must be even in order for $\langle1|\sigma^{s_1}\ldots\sigma^{s_n}|1\rangle$ to be non-zero.
Secondly,
Eq. \eqref{L13} tells us that
for even $n$  there is only one nonzero term  in the sum \eqref{L11},
namely the term
\begin{align}
A(t_1)  A'(t_2) \ldots
A(t_{n-1})  A'(t_{n}).
\end{align}
This proves
Eq. (\ref{Lr}). Eqs. (\ref{L15})-(\ref{L17}) can also be proved in the same way.
\end{proof}

\begin{prop}\label{prop1}
\begin{equation}\label{L21}
\langle1|\tilde{H}_{I_r}^1\ldots \tilde{H}_{I_r}^n|1\rangle=\sum^{n}_{k=0}(ib)^k
\overbrace{\sum_{ p_1,\ldots, p_{(n-k)/2}}A(t_{p_1}) A'(t_{p_1+1})\ldots
A(t_{p_{(n-k)/2}})A'(t_{p_{(n-k)/2}+1})}^{\mathcal{A}_k}\equiv\sum^{n}_{k=0}(ib)^k\mathcal{A}_k,
\end{equation}
where the sum runs over $1\leq p_1<p_2-1,p_2<p_3-1,\ldots,p_{(n-k)/2-1}<p_{(n-k)/2} \leq n-1$.
See Fig. (\ref{Fig2}) for an illustration.
\end{prop}
\begin{figure}[h]
\centering
\includegraphics[width=10cm]{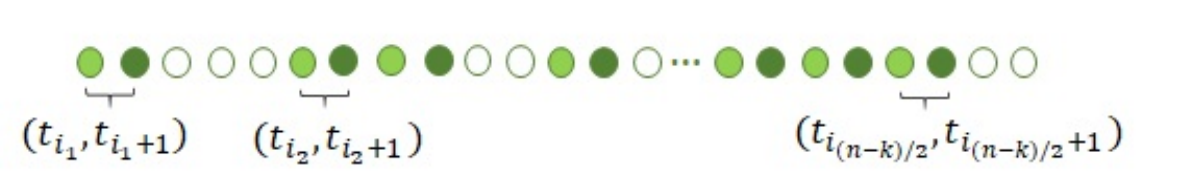}
\caption{Schematic illustration of distribution of $(ib)^k$ and the pairs $A(t_{i_j})A'(t_{i_j+1})$ for a general $n$. Each empty circle denotes one $ib$ and every two full circles denote one pair $A(t_{i_j})A'(t_{i_j+1})$. Thus the total number of the circles is $n$ and  number of the empty circles is $k$ and the number of pairs is $(n-k)/2$.}
\label{Fig2}
\end{figure}
\begin{proof}
Let us first rewrite $\tilde{H}_{I_r}$, defined in Eq. (\ref{eq:HI_A}), as
\begin{equation}\label{L20}
    \tilde{H}_{I_r}=  H_{I}+ib|1\rangle\langle1|.
\end{equation}
Then using Eq. (\ref{L20}) we have
\begin{eqnarray}\label{L22}\nonumber
\tilde{H}^1\tilde{H}^2\ldots \tilde{H}^{n-1}\tilde{H}^n&=&\sum^{n}_{j_1,j_2,j_3,\ldots j_{m}}(ib)^{(j_1+j_2+\ldots+j_{m})}\mathcal{H}_I(i_1)
|1\rangle\langle1|^{j_1}\mathcal{H}_I(i_2)
|1\rangle\langle1|^{j_2}\mathcal{H}_I(i_3)|1\rangle\langle1|^{j_3}\ldots|1\rangle\langle1|^{j_{m-1}}
\mathcal{H}_I(i_m)|1\rangle\langle1|^{j_{m}}\\\nonumber
&=&\sum^{n}_{k=0}(ib)^{k}\mathcal{H}_I(i_1)|1\rangle\langle1|
\mathcal{H}_I(i_2)|1\rangle\langle1|\mathcal{H}_I(i_3)|1\rangle\langle1|\ldots|1\rangle\langle1|
\mathcal{H}_I(i_m)|1\rangle\langle1|,\\
\end{eqnarray}
where
\begin{equation}\label{L22b}
\mathcal{H}_I(i_s)=\overbrace{H_{I}^{i_1+j_1+\ldots+i_{s-1}+j_{s-1}+s-1}\ldots H_{I}^{i_1+j_1+\ldots+i_{s-1}+j_{s-1}+i_s+s-1}}^{i_s\ times},\ \quad s=1,2,3\ldots,n\ \quad and\ \quad i_s=0,1,2,3\ldots(i_0=0),
\end{equation}
\begin{equation}\label{L22c}
|1\rangle\langle1|^{j_s}=\overbrace{|1\rangle\langle1|\ldots|1\rangle\langle1|}^{j_s\ times},\ \quad j_s=\{0,1,2,3\ldots,n\}, (j_0=0),
\end{equation}
and in the first equality $n=\sum_{s=1}^m(i_s+j_s)$ and $m$ is the number of times the pattern $\mathcal{H}_I(i_s)(ib|1\rangle\langle1|)^{j_s}$ is repeated and in the second equality $k=\sum_{s=1}^mj_s$.
Now using lemma \ref{lemma1} we have
\begin{eqnarray}\label{L22a}\nonumber
\langle1|\tilde{H}^1\ldots \tilde{H}^n|1\rangle&=&\sum^{n}_{k=0}(ib)^k
\overbrace{A(t_1)A'(t_2)\ldots A(t_{i-1})A'(t_i)}^{even}
\overbrace{A(t_{i+j+1})A'(t_{i+j+2})\ldots A(t_{l-1})A'(t_l)}^{even}\\\nonumber
&\times&\overbrace{A(t_{l+m+1})A'(t_{l+m+2})\ldots A(t_{p-1})A'(t_p)}^{even}\ldots\\
\end{eqnarray}
And since $t_i$ and $t_{i+1}$ are subsequent times in Eq. (\ref{L22a}), we can write
\begin{equation}\label{L21}
\langle1|\tilde{H}^1\ldots \tilde{H}^n|1\rangle=\sum^{n}_{k=0}(ib)^k
\sum_{i_1,\ldots,i_{(n-k)/2}}A(t_{i_1})A'(t_{i_1+1})\ldots
A(t_{i_{(n-k)/2}})A'(t_{i_{(n-k)/2}+1}),
\end{equation}
where the sum runs over $1\leq i_1<i_2-1,i_2<i_3-1,\ldots,i_{(n-k)/2-1}<i_{(n-k)/2} \leq n-1$.
\end{proof}

\begin{exmp}
As an example consider the fourth order, $n=4$, in Dyson series (see Fig. (\ref{Fig3})). We have
\begin{equation}\label{LE1}
\langle1|\tilde{H}^1\tilde{H}^2\tilde{H}^3\tilde{H}^4|1\rangle=\underbrace{A(t_{1})A'(t_{2})A(t_{3})A'(t_{4})}_{\mathcal{A}_0}
+(ib)^2\underbrace{(A(t_{1})A'(t_{2})+
A(t_{2})A'(t_{3})+A(t_{3})A'(t_{4}))}_{\mathcal{A}_2}+(ib)^4,
\end{equation}
where $\mathcal{A}_1=\mathcal{A}_3=0$ and $\mathcal{A}_4=1$.
\begin{figure}[h]
\centering
\includegraphics[width=10cm]{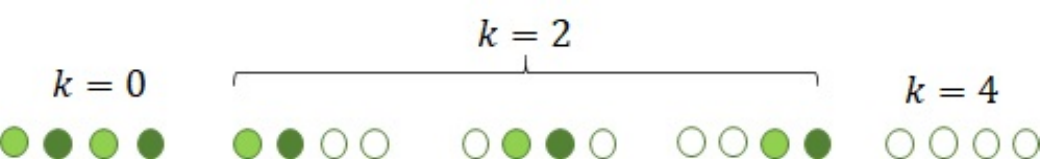}
\caption{Schematic illustration of distribution of $(ib)^k$ and the pairs $A(t_{i_j})A'(t_{i_j+1})$ for a $n=4$.}
\label{Fig3}
\end{figure}
\end{exmp}

\begin{lemma}\label{lem:commuting_a}
 For even number $p$ we have
\begin{align}
\label{LC1}
a_1a^\dagger_2\ldots a_{p-1}a^\dagger_p&=
\sum_{\substack{I,J\sim[p] \\ |I|+|J|=p/2 \\ I\cap J=\O}}
a^\dagger_{i_1+1}a^\dagger_{i_2+1}\ldots a^\dagger_{i_l+1}a_{i_1}a_{i_2}\ldots a_{i_l}
\delta(\omega_{j_1}-\omega_{j_1+1})\ldots\delta(\omega_{j_m}-\omega_{j_m+1})+\mathcal{R}(a,a^\dagger),
\nonumber \\
&\equiv
\sum_{\substack{I,J\sim[p] \\ |I|+|J|=p/2 \\ I\cap J=\O}}
:\! X_I(a,a^\dagger)\!:\, X_J(\delta) +\mathcal{R}(a,a^\dagger)
\end{align}
where $a_i\equiv a(\omega_i)$,
$I=\{(i_1,i_1+1),\ldots (i_l,i_l+1)\}$,  $J=\{(j_1,j_1+1),\ldots (j_m,j_m+1)\}$
and $l=|I|$, $m=|J|$  are the number of
pairs in $I$ and $J$, respectively. In other words we divide the set of pairs
$(1,2),(3,4),\ldots, (p-1,p)$
into two disjoint subsets $I$ and $J$, and sum up over all such possible divisions.
Moreover, $\mathcal{R}(a,a^\dagger)$ is
sum of terms of the form
\begin{equation}\label{LC2}
a^\dagger_{i'_1}a^\dagger_{i'_2}\ldots a^\dagger_{i'_l}a_{i''_1}a_{i''_2}\ldots a_{i''_{l}}
\delta(\omega_{j'_1}-\omega_{j''_1})\ldots\delta(\omega_{j'_m}-\omega_{j''_m}),
\end{equation}
where all indices $i'$, $i''$, $j'$, $j ''$ are distinct from  each other, and at least one of $\delta$'s exhibit jump, i.e. $|j'_s -j''_s|>1$ for some $s$.
\end{lemma}
\begin{proof}
This can be proven making use of Wick's theorem.
\end{proof}

\begin{exmp}
As an example consider the fourth and six orders, $n=4$ and $n=6$, in Dyson series. We have
\begin{equation}\label{LE0}
a_1a^\dagger_2a_3a^\dagger_4=a^\dagger_2a^\dagger_4a_1a_3+a^\dagger_4a_3\delta(\omega_1-\omega_2)+a^\dagger_2a_1\delta(\omega_3-\omega_4)
+\delta(\omega_1-\omega_2)\delta(\omega_3-\omega_4)+\underbrace{a^\dagger_2a_3\delta(\omega_1-\omega_4)}_{\mathcal{R}(a,a^\dagger)}.
\end{equation}
\begin{eqnarray}\label{LE1}\nonumber
a_1a^\dagger_2a_3a^\dagger_4a_5a^\dagger_6&=&a^\dagger_2a^\dagger_4a^\dagger_6a_1a_3a_5
+\delta(\omega_1-\omega_2)\delta(\omega_3-\omega_4)\delta(\omega_5-\omega_6)
+a^\dagger_2a^\dagger_4a_1a_3\delta(\omega_5-\omega_6)+a^\dagger_2a^\dagger_6a_1a_5\delta(\omega_3-\omega_4)\\\nonumber
&+&a^\dagger_4a^\dagger_6a_3a_5\delta(\omega_1-\omega_2)
+(a^\dagger_6a_5\delta(\omega_3-\omega_4)+a^\dagger_4a_3\delta(\omega_5-\omega_6))\delta(\omega_1-\omega_2)
+a^\dagger_2a_1\delta(\omega_5-\omega_6)\delta(\omega_3-\omega_4)\\\nonumber
&+&\underbrace {a^\dagger_2a^\dagger_4a_3a_5\delta(\omega_1-\omega_6)
+a^\dagger_2a^\dagger_4a_1a_5\delta(\omega_3-\omega_6)+a^\dagger_2a^\dagger_6a_3a_5\delta(\omega_1-\omega_4)
+a^\dagger_4a_5\delta(\omega_1-\omega_2)\delta(\omega_3-\omega_6)}_{\in\mathcal{R}(a,a^\dagger)}\\\nonumber
&+&\underbrace{(a^\dagger_2a_3\delta(\omega_5-\omega_6)+a^\dagger_2a_5\delta(\omega_3-\omega_6))\delta(\omega_1-\omega_4)
+a^\dagger_2a_5\delta(\omega_1-\omega_6)\delta(\omega_3-\omega_4)}_{\in\mathcal{R}(a,a^\dagger)}.\\
\end{eqnarray}
\end{exmp}

For $\mathcal{A}_k$ defined in proposition \ref{prop1} the following lemma can also be proved.
\begin{sublemma}\label{lem:commuting_A.1}
\begin{equation}\label{L26a}
\mathcal{A}_k=\sum_{\substack{I,J\sim[n] \\ |I|+|J|=(n-k)/2 \\ I\cap J=\O}} :\! X_I(A',A)\!:\,X_J(F^r)+\mathcal{R}_k(A',A,F^r),
\end{equation}\cor
where
\begin{equation}\label{L26b}
X_I(A',A)=A(t_{i_1})A'(t_{i_1+1})\ldots A(t_{i_l})A'(t_{l_l+1}),\ \quad X_J(F)=F^r(t_{j_1}-t_{j_1+1})\ldots F^r(t_{j_m}-t_{j_m+1}),
\end{equation}
The object
$\mathcal{R}_k(A,A',F)$ consists of the terms of the form
\begin{equation}\label{LC2}
A(t_{i'_1}) A'(t_{i'_2})\ldots  A'(t_{i'_l}) \, A(t_{i''_1}) A(t_{i''_2})\ldots A(t_{i''_{l}})
F^r(t_{j'_1}-t_{j''_1})\ldots F^r(t_{j'_m}-t_{j''_m}),
\end{equation}
where all indices $i'$, $i''$, $j'$, $j ''$ are distinct from  each other, and at least one of $F^r$'s exhibit jump, i.e. $|j'_s -j''_s|>1$ for some $s$.
\end{sublemma}
\begin{proof}
In lemma \ref{lem:commuting_a} we have proved the same relation for  operators $a,a^\dagger$, where we have only exploited canonical commutation relation, and the fact that $\delta(\omega_i-\omega_j)$ is a scalar.
Now, operators $A,\tilde A$ satisfy the same commutation relation, with
$F^r(t_i - t_j)$ in place of Dirac delta. Since $F$ is also scalar, we
obtain
\begin{align}\label{L27a}
&A(t_{p_1})A'(t_{p_1+1})\ldots A(t_{p_{(n-k)/2}})A'(t_{p_{(n-k)/2}+1}) \equiv X_S(A,\tilde A)=
\nonumber \\
&=\sum_{\substack{I\cup J= S  \\ I\cap J=\O}}
:\! X_I(A,\tilde A)\!:\, X_J(F^r) +\mathcal{R}_S(A,\tilde A)
,
\end{align}
Here $\mathcal{R}_S(A,A',F^r)$
has the same feature as
$\mathcal{R}(a,a^\dagger,\delta)$ (i.e. it is the sum of the terms with "jump"). Therefore we have
\begin{eqnarray}\label{L27b}\nonumber
\mathcal{A}_k&=&\sum_{p_1,\ldots,p_{(n-k)/2}} A(t_{p_1})A'(t_{p_1+1})\ldots A(t_{p_{(n-k)/2}})A'(t_{p_{(n-k)/2}+1})\equiv \sum_{\substack{S\sim [n]\\ |S|=(n-k)/2}} X_S=\\\nonumber
&=&\sum_{\substack{S\sim [n]\\ |S|=(n-k)/2}}\,\,
\sum_{\substack{I\cup J= S  \\ I\cap J=\O}}
:\! X_I(A,\tilde A)\!:\, X_J(F^r) +\mathcal{R}_S(A,\tilde A)=
 \\\nonumber
&=&\sum_{\substack{I,J\sim [n] \\ |I|+|J|=(n-k)/2}}:\!X_I(A',A)\!:\,X_J(F^r)+\mathcal{R}_k(A',A,F^r).\\
\end{eqnarray}
where $\mathcal{R}_k(A, \tilde A, F^r)=\sum_S \mathcal{R}_S$,
and by definition it has the same feature as $\mathcal{R}_S$, i.e. it is sum of terms with "jumps".
\end{proof}
\begin{exmp}
As an example, consider the fourth order, $n=4$, in Dyson series. We have
\begin{eqnarray}\label{LE2}\nonumber
\mathcal{A}_0&=&A(t_{1})A'(t_{2})A(t_{3})A'(t_{4})=\\\nonumber
&=&A'(t_{2})A'(t_{4})A(t_{1})A(t_{3})+A'(t_{2})A(t_{1})F^r(t_3-t_4)
+A'(t_{4})A(t_{3})F^r(t_1-t_2)+\underbrace{A'(t_{2})A(t_{3})F^r(t_1-t_4)}_{\mathcal{R}(A,A',F^r)}.\\
\end{eqnarray}
\end{exmp}\blk

We shall now apply the approximation introduced in Eq. \eqref{approximation}, $D(t_i-t_j)=b\delta(t_i-t_j)-F^r(t_i-t_j)\approx0$, then
\begin{equation}\label{L41}
\int dt_i\ldots\int dt_jF^r(t_i-t_j)\approx\int dt_i\ldots\int dt_jb\delta(t_i-t_j).
\end{equation}
Therefore we have the following rule
\begin{equation}\label{L42}
F^r(t_i-t_j)\ \quad \Rightarrow\ \quad b\delta(t_i-t_j).
\end{equation}
In the following the rule (\ref{L42}) will be applied.


Accordingly, here we define $\mathcal{A}^{appr}_k$ and $\mathcal{R}_k^{appr}(A,A',F^r)$ as follows.
\begin{equation}\label{L39b}
\mathcal{A}^{appr}_k\equiv\sum_{I,J}:X_I(A',A):X_J(\delta),
\end{equation}
in which the sum runs over sets $I$ of $J$ of indices, satisfying $I=\{i_1,i_1+1,\ldots, i_l,i_l+1\}$, $J=\{j_1,j_1+1\ldots, j_m,j_m+1\}$, $I\cup J=\{1,3,\ldots,(n-k)/2\}$, $I\cap J=\{\O\}$, $l+m=(n-k)/2$. $\mathcal{R}_k^{appr}(A,A',F^r)$ is like $\mathcal{R}_k(A,A',F^r)$ in Eq. (\ref{L26a}) where $F^r(t_{j_m}-t_{j_m+1})$ is replaced by $\delta(t_{j_m}-t_{j_m+1})$. Therefore using lemma \ref{lemma4} we have
\begin{equation}\label{L44}
\int dt^{n-k}\mathcal{R}_k^{appr}(A,A',F^r)=0.
\end{equation}
Now defining
\begin{equation}\label{Happ}
\langle1|\tilde{H}^1\ldots \tilde{H}^n|1\rangle^{appr}
\end{equation}

we get
\begin{equation}\label{Sapp}
\langle 1|S(0,t)|1\rangle^{appr}=\sum_{n=0}^\infty(-i)^n\int dt^{n}\langle1|\tilde{H}^1\ldots \tilde{H}^n|1\rangle^{appr}=\sum_{n=0}^\infty(-i)^n\int dt^{n}\sum^{n}_{k=0}(ib)^k\mathcal{A}^{appr}_k.
\end{equation}

Before proceeding with the above equation, here we introduce a useful diagrammatic notation for time integrals.

\subsection{Diagrams and Tableaux}
The integrand of each integral in $\langle 1|S(0,T)|1\rangle^{appr}$ in Eq. \eqref{Sapp} which are $(ib)^k\mathcal{A}^{appr}_k$ is a sequence of $ib$s, the pairs of $A'(t_{i_l+1})A(t_{i_l})$ and $b\delta(t_{j_m}-t_{j_m+1})$s. We will schematically illustrate each $ib$ by an empty circle and every pair of $A'(t_{i_l+1})A(t_{i_l})$ by two green circles and each $b\delta(t_{j_m}-t_{j_m+1})$ by two brown circles. Therefore each integrand can be represented by a sequence of such three types of objects. We will accomplish this in three steps. As an example consider the following integral whose integrand has been schematically illustrated in Fig. (\ref{Fig4}) (Step (I)):
\begin{eqnarray}\label{PE1}\nonumber
\mathcal{B}&=&\int dt^{34} b^{14}A(t_2)A'(t_3)A(t_7)A'(t_8)A(t_9)A'(t_{10})A(t_{13})A'(t_{14})A(t_{19})A'(t_{20})\\\nonumber
&\times&A(t_{21})A'(t_{22})A(t_{23})A'(t_{24})A(t_{28})A'(t_{29})A(t_{30})A'(t_{31})A(t_{33})A'(t_{34}).\\
\end{eqnarray}
After commuting the operators through many normally ordered terms are produced. For instance, one of them is the following (Step (II)):
\begin{eqnarray}\label{PE2}\nonumber
\mathcal{B}_i&=&\int dt^{34} b^{19}A'(t_3)A'(t_8)A'(t_{22})A'(t_{24})A'(t_{34})A(t_2)A(t_7)A(t_{21})A(t_{23})A(t_{33})\\\nonumber
&\times&\delta(t_{9}-t_{10})\delta(t_{13}-t_{14})\delta(t_{19}-t_{20})\delta(t_{28}-t_{29})\delta(t_{30}-t_{31}).\\
\end{eqnarray}
And then from lemma \ref{lemma4.1} we know that any $\delta(t_j-t_{j+1})$ reduces the time integral by one. Hence (Step (III))
\begin{equation}\label{PE3}
\mathcal{B}_i
=\int dt^{29} b^{19}A'(t_3)A'(t_8)A'(t_{19})A'(t_{21})A'(t_{29})A(t_2)A(t_7)A(t_{18})A(t_{20})A(t_{28}).
\end{equation}
Then we encode the positions of these three types of circles in the \textit{final} distribution, i.e., in Step (III) by a diagram $\mu$ and a tableau $y$ (see Fig. (\ref{Fig5})). The diagram $\mu$ is used to show the general form of the integrand, i.e., the distribution of $b$s and the pairs $A'A$ and the tableau $y$, which is the diagram $\mu$ filled with the sequences of 0s (for $ib$s) and 1s (for $b\delta(t_j-t_{j+1})$s), demonstrates the specific form of the integrand which means that it precisely encodes the positions of $ib$s and $b\delta(t_j-t_{j+1})$s. The number of blocks in each row $|\mu_j|$ shows the final number of time integrals in Step (III). After each row there exists a pair of $A'(t_i)A(t_{i+1})$ and $P(\mu)=(|\mu_1|+1, |\mu_1|+|\mu_2|+3, |\mu_1|+|\mu_2|+|\mu_3|+5,\ldots,\sum_{j=1}^{l}|\mu_j|+2l-1)$ completely determines the position of $A(t_i)$ (hence the position of the pair). Note that if a row is empty, this means that the pairs $A'(t_{i+1})A(t_{i})$ come one after another. It is seen that for $l=5$ pairs of $A'A$ there will be $l+1=6$ rows in diagram $\mu=(\mu_1,\mu_2,\ldots, \mu_{l+1})$ (see Fig. (\ref{Fig5})). In this example we have $P(\mu)=(1+1, 1+3+3, 1+3+9+5,1+3+9+0+7,1+3+9+0+6+9,1+3+9+0+6+0+11)=(2, 7, 18,20,28)$.
\begin{figure}[h]
\centering
\includegraphics[width=10cm]{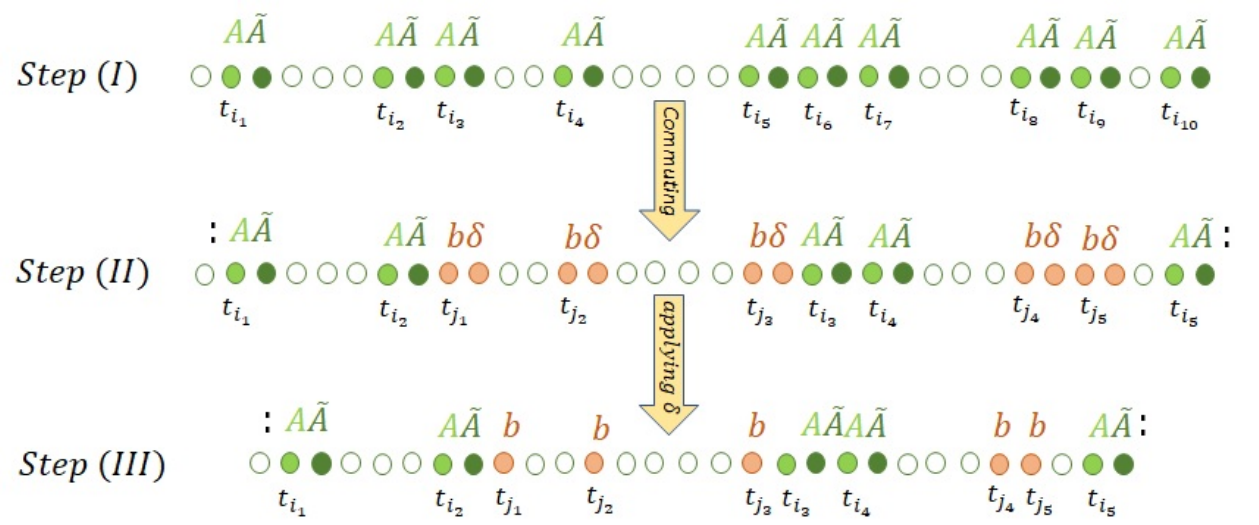}
\caption{Schematic illustration of distribution of the integrand containing $(ib)^k$ (empty circles), $b\delta(t_{j_m}-t_{j_m+1})$ (every two brown circles), and the pairs $A(t_{i_l})A'(t_{i_l+1})$ (every two green circles) for $n=34$ and $k=14$. Step (I). The original integrand. Step (II). Commuting the operators through. Step (III). Applying $\delta(t_{j_m}-t_{j_m+1})$.'$:x:$' denotes normal ordering.}
\label{Fig4}
\end{figure}
\begin{figure}[h]
\centering
\includegraphics[width=10cm]{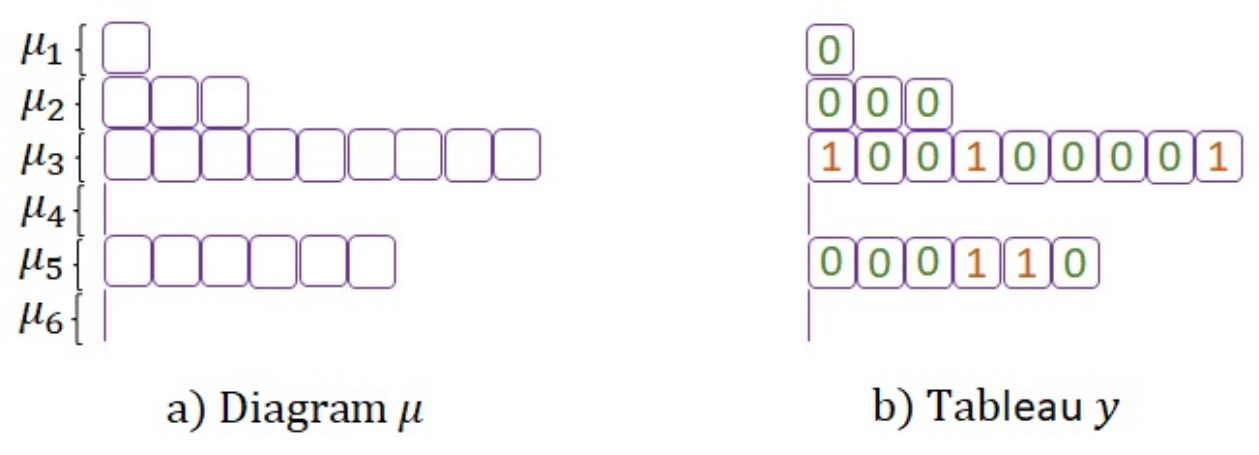}
\caption{For $l$ pairs of $A'A$ there will be $l+1$ rows in diagram $\mu=(\mu_1,\mu_2,\ldots, \mu_{l+1})$. a) The diagram $\mu$ shows how $b$s and the pairs $A'A$ are distributed in Fig. (\ref{Fig4}). As can be seen $P(\mu)=(2, 7, 18,20,28)$ completely determines the position of $A(t_{i_l})$ in Step (III) of Fig. (\ref{Fig4}). b) The tableau $y$, which is the diagram $\mu$ filled with the sequences of 0s and 1s, specifies the positions of $ib$s and $b\delta$s in Step (III) of Fig. (\ref{Fig4}).}
\label{Fig5}
\end{figure}

\subsection{Proof of the main result}

\begin{prop}\label{prop2}
\begin{equation}\label{P1}
\langle 1|S_{I}(t,0)|1\rangle^{appr}=\sum_{l=0}^\infty(-1)^l\sum_{\mu\vdash l+1}b
^{|\mu|}\int dt^{|\mu|+2l}:\!X_{P(\mu)}(A,A')\!:\sum_{y\in\mu}(-1)^{wt(y)},
\end{equation}
where
\begin{align}
P(\mu)=\{(P_1,P_1+1),\ldots,(P_l,P_l+1)\},\quad \text{with} \quad  P_r=\sum_{j=1}^{r}|\mu_j|+2r-1,\quad r=1,\ldots,l
\end{align}
The second sum runs over all $\mu$ with $l+1$ rows (denoted here by $\mu\vdash l+1$) and $y\in\mu$ means that $y$ is $\mu$  filled with sequences of 0s and 1s showing the positions of $ib$s and $F$s after commuting the the pairs of $A(t_{i_l})A'(t_{i_l+1})$. The index $P(\mu)$  shows the position of the pairs (see Figs. (\ref{Fig4}) and (\ref{Fig5}) for more detail). $|\mu|$ is the number of blocks in $\mu$ and $w(y)$ is the number of 1s in $y$.
\end{prop}
\begin{proof}\cor
Substituting Eq. (\ref{L39b}) into Eq. (\ref{Sapp}) we get
\begin{eqnarray}\label{P2}\nonumber
\langle 1|S_{I}(t,0)|1\rangle^{appr}&=&\sum_{n=0}^\infty(-i)^n\int dt^{n}\sum^{n}_{k=0}(ib)^k\mathcal{A}^{appr}_k\\\nonumber
&=&\sum_{n=0}^\infty\int dt^{n}\sum^{n}_{k=0}(-i)^{n-k}b^k\sum_{\substack{I,J\sim[n],\\I\cap J=\O,\\|I|+|J|=(n-k)/2}}:\!X_I(A,A')\!:\,X_J(\delta)=\\\nonumber
&=&\sum_{n=0}^\infty\sum_{k,l,m}(-1)
^{l+m}b^{k+m}\sum_{\substack{I,J\sim[n],\\I\cap J=\O,\\|I|=l,|J|=m}}\int dt^{n}:\!X_I(A,A')\!:\,X_J(\delta),\\
\end{eqnarray}
where the second sum in the last equality runs over all $k$ such that $n-k$ is even and $l, m\geq0$ such that $m+l=(n-k)/2$.
where $X_I(A',A)$ and $X_J(\delta)$ were introduced in Eqs. (\ref{L26b}) and (\ref{L39b}), respectively, and $":x:"$ denotes normal ordering. Now Eq. (\ref{P2}) may be rewritten as
\begin{align}
\langle 1|S_{I}(t,0)|1\rangle^{appr}=
\sum_{k,l,m=0}^\infty(-1)
^{l+m}b^{k+m}\sum_{\substack{I,J\sim \mathbb{N},\\I\cap J=\O,\\|I|=l,|J|=m}}\int dt^{m+l+2k}:\!X_I(A,A')\!:\,X_J(\delta).\\
\end{align}
We can now further rewrite it using diagrams:
\begin{equation}\label{P3}
\langle 1|S_{I}(t,0)|1\rangle^{appr}=\sum_{l=0}^\infty(-1)^l\sum_{\mu\vdash l+1}b
^{|\mu|}\sum_{y\in\mu}(-1)^{wt(y)}\int dt^{|\mu|+2l}:\!X_{P(\mu)}(A,A')\!:,
\end{equation}
where the second sum runs over all $\mu$ with $l+1$ rows (denoted here by $\mu\vdash l+1$) and $y\in\mu$ means that $y$ is $\mu$ already filled with sequences of 0s and 1s showing the positions of $ib$s and $F$s commuting the pairs of $A(t_{i_l})A'(t_{i_l+1})$ and $P(\mu)$
for $\mu$ with $l+1$ rows is given by
\begin{align}
P(\mu)=\{(P_1,P_1+1),\ldots,(P_l,P_l+1)\},\quad \text{with} \quad  P_r=\sum_{j=1}^{r}|\mu_j|+2r-1.
\end{align}
Here $|\mu|$ is the number of blocks in $\mu$ and $wt(y)$ is the number of 1s in $y$.
Since $X_{P(\mu)}$ depends only on $\mu$ and not on $y$,  Eq. (\ref{P2}) may be rewritten as
\begin{equation}\label{P4}
\langle 1|S_I(t,0)|1\rangle^{appr}=\sum_{l=0}^\infty(-1)^l\sum_{\mu\vdash l+1}b
^{|\mu|}\int dt^{|\mu|+2l}:\!X_{P(\mu)}(A,A')\!:\sum_{y\in\mu}(-1)^{w(y)}.
\end{equation}
\end{proof}

\begin{lemma}\label{lemma5}
For all $\mu$ apart from empty one $|\mu|=0$ we have
\begin{equation}\label{L45}
\sum_{y\in\mu}(-1)^{w(y)}=0.
\end{equation}
\end{lemma}
\begin{proof}
\begin{eqnarray}\label{L46}\nonumber
\sum_{y\in\mu}(-1)^{w(y)}&=&\sum_{m=0}^{|\mu|}(-1)^m\dfrac{(|\mu|+l)!}{l!m!(|\mu|-m)!}\\\nonumber
&=&\dfrac{(|\mu|+l)(|\mu|+l-1)\ldots(|\mu|+1)}{l!}\sum_{m=0}^{|\mu|}(-1)^m\dfrac{|\mu|!}{m!(|\mu|-m)!}\\\nonumber
&=&\dfrac{(|\mu|+l)(|\mu|+l-1)\ldots(|\mu|+1)}{l!}\sum_{m=0}^{|\mu|}(-1)^m\left(\begin{array}{l}
|\mu| \\
m
\end{array}\right).\\
\end{eqnarray}
From the Binomial theorem \cite{Rotman} we know that if $x$ and $y$ are variables and $n\in \mathds{N}$, then
\begin{equation}\label{Bi1}
(x+y)^{n}=\sum_{k=0}^{n}\left(\begin{array}{l}
n \\
k
\end{array}\right) x^{n-k} y^{k}.
\end{equation}
Now choosing $x=1$ and $y=-1$ we get
\begin{equation}\label{Bi2}
0=(1+(-1))^{n}=\sum_{k=0}^{n}\left(\begin{array}{l}
n \\
k
\end{array}\right) 1^{n-k}(-1)^{k}=\sum_{k=0}^{n}\left(\begin{array}{l}
n \\
k
\end{array}\right) (-1)^{k}=\sum_{m=0}^{|\mu|}(-1)^m\left(\begin{array}{l}
|\mu| \\
m
\end{array}\right).
\end{equation}
\end{proof}
We are now in a position to prove a proposition, that gives us Eq. \eqref{th1} from theorem \ref{thm1}. The other equations from this theorem are obtained analogously.

\begin{prop}
\begin{align}
\langle 1|S_{I}(t,0)|1\rangle^{appr}=
\sum_{l=0}^\infty
(-i)^{2l} \int dt^{2l}
 A'(t_2) A'(t_4) \ldots A'(t_{2l})\, A(t_1) A(t_3)\ldots A(t_{2l-1}).
\end{align}
\end{prop}

\begin{proof}

Inserting Eq. \eqref{L45}
from lemma   	 \ref{lemma5}
into Eq. \eqref{P1} we obtain
the required result.
We see that only one diagram is left in the sum over $\mu$ - the trivial one - with no blocks (let us call it $\mu_{null}$).
This diagram has very simple $P(\mu_{null})=(1,2),(3,4),\ldots (2l-1,2l)$.
This gives
\begin{align}
\langle1|S_{I}(t,0)|1\rangle^{appr}=
\sum_{l=0}^\infty
(-i)^{2l} :\! X_P(\mu_{null})(A,A'):,
\end{align}
which by definition of $X$ gives the required result.
\end{proof}

\subsection{Example}
\begin{exmp}
As an example consider the term $b^2a_1a_2^\dagger a_3a_4^\dagger a_5a_6^\dagger$ in the eighth order of Dyson series. In Fig. (\ref{Fig6}) it is illustrated how $(ib)^2$, the pairs $A'A$ and $b\delta$s are distributed after commuting the operators through. Then taking the integral over the time we get the following terms:
\begin{figure}[h]
\centering
\includegraphics[width=16cm]{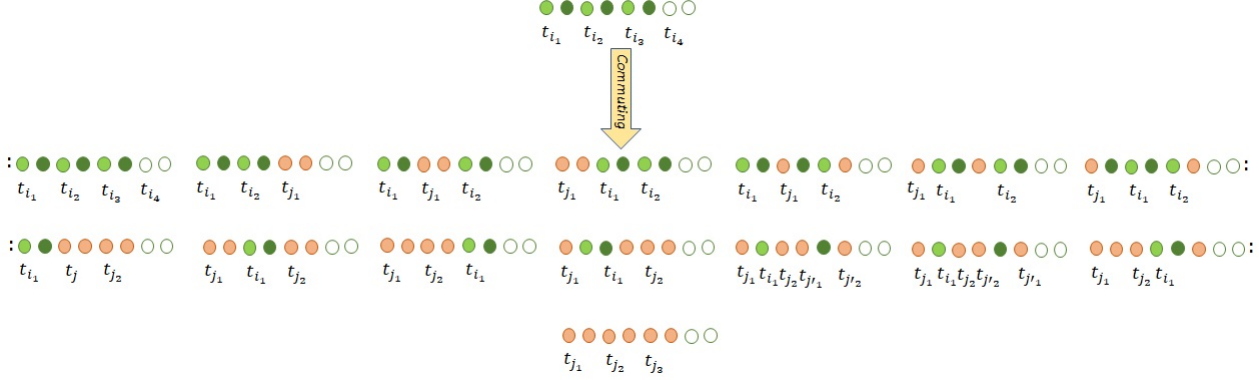}
\caption{Schematic illustration of distribution of $(ib)^2$ (empty circles), $b\delta(t_{j_m}-t_{j_m+1})$ (every two brown circles), and the pairs $A(t_{i_l})A'(t_{i_l+1})$ (every two green circles) for the term $b^2A(t_{1})A'(t_{2})A(t_{3})A'(t_{4})A(t_{5})A'(t_{6})$ (steps (I) and (II)). For step (III) from every two brown circles one of them is removed. '$:x:$' denotes normal ordering.}
\label{Fig6}
\end{figure}
\begin{equation}\label{L47}
\int dt^8b^2A'(t_{2})A'(t_{4})A'(t_{6})A(t_{1})A(t_{3})A(t_{5})=\int dt^8b^2:X_{P_1,P_2,P_3}:,
\end{equation}
where $X_{P_1,P_2,P_3}=A'(t_{2})A'(t_{4})A'(t_{6})A(t_{1})A(t_{3})A(t_{5})$, $P_1=\mu_1+2\times1-1=0+2-1=1$, $P_2=\mu_1+\mu_2+2\times2-1=0+0+4-1=3$ and $P_3=\mu_1+\mu_2+\mu_3+2\times3-1=0+0+0+6-1=5$ (see Fig. (\ref{Fig7}) for more detail).
\begin{equation}\label{L48}
\int dt^8b^2A'(t_{2})A'(t_{4})A(t_{1})A(t_{3})F(t_{5}-t_{6})\simeq\int dt^8b^2A'(t_{2})A'(t_{4})A(t_{1})A(t_{3})\delta(t_{5}-t_{6})
=\int dt^7b^2:X_{P_1,P_2}:,
\end{equation}
where $X_{P_1,P_2}=A'(t_{2})A'(t_{4})A(t_{1})A(t_{3})$, $P_1=\mu_1+2\times1-1=0+2-1=1$ and $P_2=\mu_1+\mu_2+2\times2-1=0+0+4-1=3$.
\begin{equation}\label{L49}
\int dt^8b^2A'(t_{2})A'(t_{6})A(t_{1})A(t_{5})F(t_{3}-t_{4})\simeq\int dt^8b^2A'(t_{2})A'(t_{6})A(t_{1})A(t_{5})\delta(t_{3}-t_{4})
=\int dt^7b^2:X_{P_1,P_2}:,
\end{equation}
where $X_{P_1,P_2}=A'(t_{2})A'(t_{5})A(t_{1})A(t_{4})$, $P_1=\mu_1+2\times1-1=0+2-1=1$ and $P_2=\mu_1+\mu_2+2\times2-1=0+1+4-1=4$.
\begin{equation}\label{L50}
\int dt^8b^2A'(t_{4})A'(t_{6})A(t_{3})A(t_{5})F(t_{1}-t_{2})\simeq\int dt^8b^2A'(t_{4})A'(t_{6})A(t_{3})A(t_{5})\delta(t_{1}-t_{2})
=\int dt^7b^2:X_{P_1,P_2}:,
\end{equation}
where $X_{P_1,P_2}=A'(t_{3})A'(t_{5})A(t_{2})A(t_{4})$, $P_1=\mu_1+2\times1-1=1+2-1=2$ and $P_2=\mu_1+\mu_2+2\times2-1=1+0+4-1=4$.
\begin{equation}\label{L51}
\int dt^8b^2A'(t_{2})A(t_{1})F(t_{3}-t_{4})F(t_{5}-t_{6})\simeq\int dt^8b^2A'(t_{2})A(t_{1})\delta(t_{3}-t_{4})\delta(t_{5}-t_{6})
=\int dt^6b^2:X_{P_1}:,
\end{equation}
where $X_{P_1}=A'(t_{2})A(t_{1})$ and $P_1=\mu_1+2\times1-1=0+2-1=1$.
\begin{equation}\label{L52}
\int dt^8b^2A'(t_{4})A(t_{3})F(t_{1}-t_{2})F(t_{5}-t_{6})\simeq\int dt^8b^2A'(t_{4})A(t_{3})\delta(t_{1}-t_{2})\delta(t_{5}-t_{6})
=\int dt^6b^2:X_{P_1}:,
\end{equation}
where $X_{P_1}=A'(t_{3})A(t_{2})$ and $P_1=\mu_1+2\times1-1=1+2-1=2$.
\begin{equation}\label{L53}
\int dt^8b^2A'(t_{6})A(t_{5})F(t_{1}-t_{2})F(t_{3}-t_{4})\simeq\int dt^8b^2A'(t_{6})A(t_{5})\delta(t_{1}-t_{2})\delta(t_{3}-t_{4})
=\int dt^6b^2:X_{P_1}:,
\end{equation}
where $X_{P_1}=A'(t_{4})A(t_{3})$ and $P_1=\mu_1+2\times1-1=2+2-1=3$.
\begin{equation}\label{L54}
\int dt^8b^2F(t_{1}-t_{2})F(t_{3}-t_{4})F(t_{5}-t_{6})\simeq\int dt^8b^2\delta(t_{1}-t_{2})\delta(t_{3}-t_{4})\delta(t_{5}-t_{6})
=\int dt^5b^2.
\end{equation}
\begin{figure}[h]
\centering
\includegraphics[width=8cm]{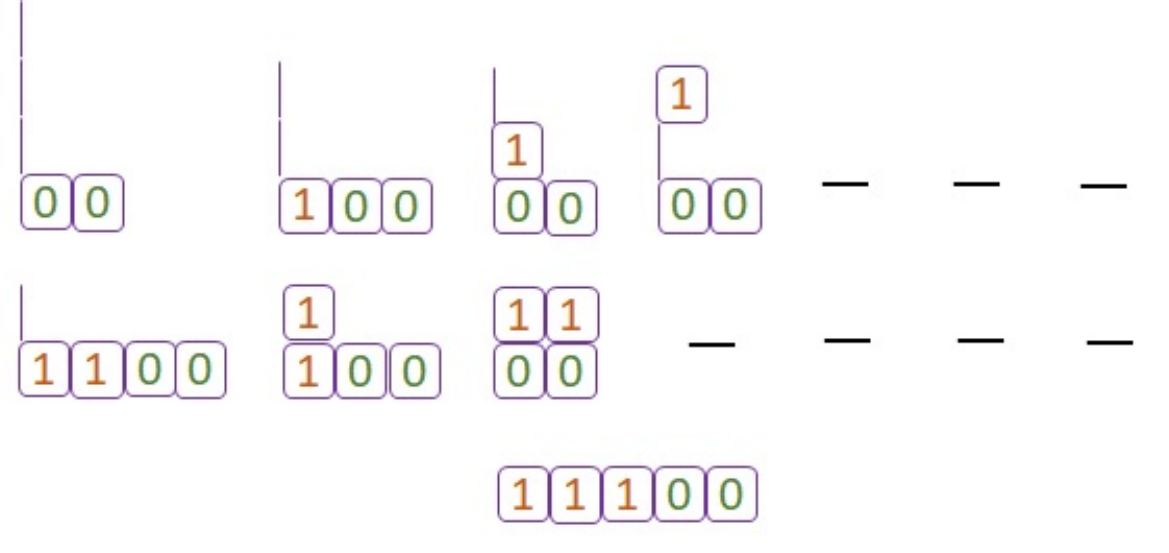}
\caption{The tableau $y$ for the term $b^2A(t_{1})A'(t_{2})A(t_{3})A'(t_{4})A(t_{5})A'(t_{6})$. The dash "-" indicates that this term lies in $\mathcal{R}(A,A')$ hence its integral over time vanishes.}
\label{Fig7}
\end{figure}
\end{exmp}

\subsection{Auxiliary lemmas}
\begin{lemma}\label{lem:commuting_A}
\begin{equation}\label{CL1-lem}
[A(t_{i}),A'(t_{j})]=F^r(t_{i}-t_{j}),
\end{equation}
where
\begin{equation}\label{LF}
F^r(t)= e^{i(\omega_A-i\gamma)t}F(t),\quad F(t)\equiv\int dk\, e^{-i\omega_kt} |f(\omega_k)|^2.
\end{equation}
\end{lemma}
\begin{proof}
Using Eq. (\ref{AA}) we have
\begin{eqnarray}\label{CL2}\nonumber[A(t_{i}),A'(t_{j})]&=&\left [\int dk\ f(\omega_k,t_i)a_k,\int dk'\ f'(\omega_{k'},t_j)a^\dagger_{k'}\right]\\\nonumber
&=&\int dk\int dk'\ f(\omega_k,t_i)f'(\omega_{k'},t_j)[a_k,a^\dagger_{k'}]\\\nonumber
&=&\int dk\int dk'f(\omega_k,t_i)f'(\omega_{k'},t_j)\delta(\omega_k-\omega_{k'})\\\nonumber
&=&\int dk\ e^{i(\omega_A-i\gamma)(t_i-t_j)}e^{-i\omega_k(t_i-t_j)}|f(\omega_k)|^2\\\nonumber
&=&e^{i(\omega_A-i\gamma)(t_i-t_j)}\int dk\ e^{-i\omega_k(t_i-t_j)}|f(\omega_k)|^2\\\nonumber
&=&F^r(t_{i}-t_{j}),\\
\end{eqnarray}
where $F^r(t)$ is given by Eq. \eqref{Frt} as
\begin{align}\label{Ft}
F^r(t)\equiv e^{i(\omega_A-i\gamma)t}\int dk\, e^{-i\omega_kt} |f(\omega_k)|^2.
\end{align}
\end{proof}
\begin{lemma}\label{lemma4}
\begin{equation}\label{L35}
\int_{0}^{t}dt_1\ldots\int_{0}^{t_{i-1}}dt_i\ldots\int_{0}^{t_{j-1}}dt_j\ldots\int_{0}^{t_{k-1}}dt_kZ(t_1,\ldots,t_k)\delta(t_i-t_j)=0,\ \quad for\ \quad j>i+1,
\end{equation}
where $Z(t_1,\ldots,t_k)$ is an arbitrary function over $t_1,\ldots,t_k$.
\end{lemma}
\begin{proof}
We use the change of variables as $u=t_i-t_j$ then
\begin{equation}\label{L32}
du=-dt_j,\ \quad t_j=0:\ u=t_i,\ \quad and\ \quad t_j=t_{j-1}:\ u=t_i-t_{j-1},
\end{equation}
so we keep $t_i$ and $u$. Now we define
\begin{eqnarray}\label{L33}\nonumber
K(t_1,\ldots,t_{j-1},t_{j+1},\ldots,t_k)&=:&\int_{0}^{t_{j-1}}dt_j\ Z(t_1,\ldots,t_k)\delta(t_i-t_j)\\\nonumber
&=&\int_{t_i-t_{j-1}}^{t_i}du\ \delta(u)Z(t_1,\ldots,t_i,\ldots,t_{j-1},t_i-u,t_{j+1},\ldots,t_k)=\left\{\begin{array}{cc}
{1} & {\text {\textit{for} $t_i=t_{j-1}$,}} \\
{0} & {\text {\textit{otherwise}.}}
\end{array}\right.\\
\end{eqnarray}
Since the function $K(t_1,\ldots,t_{j-1},t_{j+1},\ldots,t_k)$ is only nonzero at one point then its integral over $t_i$ and $t_{j-1}$ is zero, thus
\begin{equation}\label{L34}
\int_{0}^{t}dt_1\ldots\int_{0}^{t_{i-1}}dt_i\ldots\int_{0}^{t_{j-2}}dt_{j-1}K(t_1,\ldots,t_{j-1},t_{j+1},\ldots,t_k)=0,
\end{equation}
which completes the proof.
\end{proof}
\begin{sublemma}\label{lemma4.1}
\begin{equation}\label{L36}
\int_{0}^{t}dt_1\int_{0}^{t_1}dt_2\ \delta(t_1-t_2)=\int_{0}^{t}dt_1,
\end{equation}
where $t_1$ and $t_2$ are subsequent times in Dyson series.
\end{sublemma}
\begin{proof}
We use the following change of variables: $u=t_1-t_2$ then
\begin{equation}\label{L37}
du=-dt_2,\ \quad t_2=0:\ u=t_1,\ \quad and\ \quad t_2=t_1:\ u=0.
\end{equation}
So keeping $u$ and $t_1$ we have
\begin{equation}\label{L38}
\int_{0}^{t}dt_1\int_{0}^{t_1}du\ \delta(u)=\int_{0}^{t}dt_1,
\end{equation}
which completes the proof.
\end{proof}
\begin{remark}\label{remark1}
Lemma \ref{lemma4} states that the Dirac delta $\delta(t_i-t_j)$ with times that are not neighboring causes the expression to vanish and lemma \ref{lemma4.1} states that, for neighboring times, the Dirac delta $\delta(t_i-t_j)$ removes one integral.
\end{remark}
\section{Proof of Eq. \eqref{Nm62} i.e. checking normalization of the total state for initial vacuum}\label{PE62}
First let us see if the re-normalized time-evolution operator given in Eq. \eqref{Uint} preserves the norm of the state vector. Using Eqs. \eqref{AA} and \eqref{th1}-\eqref{th4} for the initial states $|\{0\}\rangle$ and $|\{1_k\}\rangle$ (one photon in each mode) of the field, respectively, we get
\begin{eqnarray}\label{ٔNorm0}\nonumber
S_{I_r}(t,0)|1,\{0\}\rangle&=&S_{I_r,11}(t,0)|1,\{0\}\rangle+S_{I_r,01}(t,0)|0,\{0\}\rangle\\\nonumber
&=&|1,\{0\}\rangle-i\int_{0}^{t}dt_1A'(t_1) |0,\{0\}\rangle\\\nonumber
&=&|1,\{0\}\rangle-i\int_{0}^{t} dt_1\int dk\ f'(\omega_k,t_1)|0,\{1_k\}\rangle.\\
\end{eqnarray}
For Schr\"{o}dinger picture, since $U(t) = e^{-iH^r_0 t}  U_{I_r}(t)$ using Eq. \eqref{Sth} we have
\begin{eqnarray}\label{Norm4}\nonumber
S(t,0)|1,\{0\}\rangle&=&e^{-iH_{0_r} t}S_{I_r}(t,0)|1,\{0\}\rangle\\\nonumber
&=&e^{-\gamma t|1\rangle\langle1|} V S_{I_r}(t,0)|1,\{0\}\rangle\\\nonumber
&=&V\left(e^{-\gamma t}|1,\{0\}\rangle-i\int_{0}^{t} dt_1\int dk\ f'(\omega_k,t_1)|0,\{1_k\}\rangle\right),\\
\end{eqnarray}
where
\begin{equation}\label{Norm5}
V=e^{-i(\omega_A |1\rangle\langle1|+H_F)t}
\end{equation}
is a unitary operator and $H_F$ is defined in Eq. \eqref{H0r}. Therefore using Eq. \eqref{Norm4},  we get
\begin{equation}\label{ٔNorm7a}
\langle \{0\},1|S^\dagger(t,0)S(t,0)|1,\{0\}\rangle=e^{-2\gamma t}+\mathcal{D}(t,0)\approx1,
\end{equation}
where $\mathcal{D}(t,0)$ is obtained as
\begin{eqnarray}\label{ٔNorm1a}\nonumber
\mathcal{D}(t,0)&=&\int dk\ |f(\omega_k)|^2\int_{0}^{t} dt_1e^{-\gamma t_1}e^{-i(\omega_k-\omega_A)t_1}\int_{0}^{t}dt'_1e^{-\gamma t'_1}e^{i(\omega_k-\omega_A)t'_1}\\\nonumber
&=&\int dk\ |f(\omega_k)|^2\left(\dfrac{e^{-\gamma t}e^{-i(\omega_k-\omega_A)t}-1}{-\gamma-i(\omega_k-\omega_A)}\right)\left(\dfrac{e^{-\gamma t}e^{i(\omega_k-\omega_A)t}-1}{-\gamma+i(\omega_k-\omega_A)}\right)\\\nonumber
&=&\int dk\ |f(\omega_k)|^2\dfrac{e^{-2\gamma t}+1-e^{-\gamma t}(e^{i(\omega_k-\omega_A)t}+e^{-i(\omega_k-\omega_A)t})}{\gamma^2+(\omega_k-\omega_A)^2}\\\nonumber
&=&(e^{-2\gamma t}+1)\int dk\ \dfrac{|f(\omega_k)|^2}{\gamma^2+(\omega_k-\omega_A)^2}-e^{-\gamma t}\int dk\ |f(\omega_k)|^2\dfrac{e^{i(\omega_k-\omega_A)t}+e^{-i(\omega_k-\omega_A)t}}{\gamma^2+(\omega_k-\omega_A)^2}\\\nonumber
&=&(e^{-2\gamma t}+1)\int d\omega\ \dfrac{4\pi\omega^2|f(\omega)|^2}{\gamma^2+(\omega-\omega_A)^2}-e^{-\gamma t}\int d\omega\ 4\pi\omega^2|f(\omega)|^2\dfrac{e^{i(\omega-\omega_A)t}+e^{-i(\omega-\omega_A)t}}{\gamma^2+(\omega-\omega_A)^2}.\\
\end{eqnarray}
Now we assume that $4\pi\omega^2|f(\omega)|^2$ is a slow varying function such that under the Lorentzian distribution $\dfrac{\gamma}{\gamma^2+(\omega-\omega_A)^2}$ we have $4\pi\omega^2|f(\omega)|^2\approx4\pi\omega_A^2|f(\omega_A)|^2$. Therefore we get
\begin{eqnarray}\label{ٔNorm1aa}\nonumber
\mathcal{D}(t,0)&\approx&\dfrac{4\pi\omega_A^2|f(\omega_A)|^2}{\gamma}\left(\pi(e^{-2\gamma t}+1)- e^{-\gamma t}\int d\omega\ \dfrac{\gamma}{\gamma^2+(\omega-\omega_A)^2}(e^{i(\omega-\omega_A)t}+e^{-i(\omega-\omega_A)t})\right)\\\nonumber
&=&\dfrac{4\pi\omega_A^2|f(\omega_A)|^2}{\gamma}\left(\pi(e^{-2\gamma t}+1)- e^{-\gamma t}\left(e^{-i\omega_At}\int d\omega\ \dfrac{\gamma e^{i\omega t}}{\gamma^2+(\omega-\omega_A)^2}+e^{i\omega_At}\int d\omega\ \dfrac{\gamma e^{-i\omega t}}{\gamma^2+(\omega-\omega_A)^2}\right)\right)\\\nonumber
&=&\dfrac{4\pi\omega_A^2|f(\omega_A)|^2}{\gamma}\left(\pi(e^{-2\gamma t}+1)-e^{-\gamma t}(e^{-i\omega_At}\pi e^{-\gamma t}e^{i\omega_At}+e^{i\omega_At}\pi e^{-\gamma t}e^{-i\omega_At})\right)\\\nonumber
&=&\dfrac{4\pi^2\omega_A^2|f(\omega_A)|^2}{\gamma}\left(e^{-2\gamma t}+1-2e^{-2\gamma t}\right)\\\nonumber
&=&1-e^{-2\gamma t},\\
\end{eqnarray}
in which we used the relation for Fourier transform of Lorentzian distribution
\begin{equation}\label{Norm1b}
\int_{-\infty}^{\infty}dx\ \dfrac{c}{c^2+(x-a)^2}e^{\pm ixt}=\pi e^{-c|t|}e^{\pm iat},
\end{equation}
and in the last equality of Eq. \eqref{ٔNorm1aa}  we used Eq. \eqref{DP3}.

\section{Expanding $\Lambda(t)$ up to the 2nd order}\label{2ndorder}

In order to confirm our conjecture let us now expand $\tilde{\rho}_A(t)$ up to the second order. First, let us calculate \eqref{nv4} with all the details as follows 
\begin{align*}\label{1st order}
\rho_A(t)&= \Lambda (t)\rho_A(0)= \int\mathcal{D}\xi\,e^{-\|\xi\|^2}U^{\xi}_{\alpha}(t) \rho_A(0)\bigl(U^{\xi}_{\alpha}(t)\bigr)^{\dagger}= \int\mathcal{D}\xi\,e^{-\|\xi\|^2} U_{\alpha}(t,0)\Big[\mathbb{I}-\int_0^t ds_1\Big(\frac{\gamma}{2}\tilde{P}_1(s_1)+if^*_{\xi}(s_1)\tilde{\sigma}^-(s_1)\Big]\Big)\rho_A(0)\\\nonumber
&\Big[\mathbb{I}-\int_0^t ds_2\Big(\frac{\gamma}{2}\tilde{P}_1(s_2)-if_{\xi}(s_2)\tilde{\sigma}^+(s_2)\Big)\Big] U_{\alpha}^\dagger(t,0)= U_{\alpha}(t,0)\Big[\rho_A(0) -\int_0^t ds_1\frac{\gamma}{2}\left (\tilde{P}_1(s_1)\rho_A(0)+\rho_A(0)\tilde{P}_1(s_1)\right)\\\nonumber
&+ \int\mathcal{D}\xi\,e^{-\|\xi\|^2} \int_0^t ds_1\int_0^t ds_2f^*_{\xi}(s_1)f_{\xi}(s_2)\tilde{\sigma}^-(s_1)\rho_A(0) \tilde{\sigma}^+(s_2)\Big]U_{\alpha}^\dagger(t,0)= U_{\alpha}(t,0)\Big[\rho_A(0) -\frac{\gamma}{2}\int_0^t ds_1\Big(\tilde{P}_1(s_1)\rho_A(0) \\\nonumber
&+\rho_A(0)\tilde{P}_1(s_1)\Big)+ \int_0^t ds_1\int_0^t ds_2\ e^{\gamma( s_1-s_2)}e^{-i\omega_A(s_1-s_2)}F^r(s_1-s_2)
\tilde{\sigma}^-(s_1)\rho_A(0) \tilde{\sigma}^+(s_2)\Big]U_{\alpha}^\dagger(t,0)\\\nonumber
&\approx U_{\alpha}(t,0)\Big[\rho_A(0) -\frac{\gamma}{2}\int_0^t ds_1\left(\tilde{P}_1(s_1)\rho_A(0)+\rho_A(0)\tilde{P}_1(s_1)\right)\\ \nonumber
&+ \int_0^t ds_1\int_0^{s_1} ds_2\ \gamma e^{\gamma (s_1-s_2)}e^{-i\omega_A(s_1-s_2)}\delta(s_1-s_2) \tilde{\sigma}^-(s_1)\rho_A(0) \tilde{\sigma}^+(s_2)\Big]U_{\alpha}^\dagger(t,0)\\\nonumber
&= U_{\alpha}(t,0)\Big[\rho_A(0) - \frac{\gamma}{2}\int_0^t ds_1\left(\tilde{P}_1(s_1)\rho_A(0)+\rho_A(0)\tilde{P}_1(s_1)\right)+ \gamma\int_0^t ds_1\  \tilde{\sigma}^-(s_1)\rho_A(0)\tilde{\sigma}^+(s_1)\Big]U_{\alpha}^\dagger(t,0).\numberthis
\end{align*}
For the second order expansion we use Eq. \eqref{nv3} and we get 
\begin{align*}\label{nv7}
\rho_A(t)&= \Lambda (t){\rho}_A^{\alpha}(0)
= \int\mathcal{D}\xi\,e^{-\|\xi\|^2} U^{\xi}_{\alpha}(t) \rho_A(0)\bigl(U^{\xi}_{\alpha}(t)\bigr)^{\dagger}
= U_{\alpha}(t,0)\Big[{\rho}_A^{\alpha}(0) -\int_0^t ds_1\frac{\gamma}{2}\Big(\tilde{P}_1(s_1){\rho}_A^{\alpha}(0)+{\rho}_A^{\alpha}(0)\tilde{P}_1(s_1)\Big) \\ \nonumber
&- \int\mathcal{D}\xi\,e^{-\|\xi\|^2} \int_0^t ds_1\int_0^t ds_2f^*_{\xi}(s_1)f_{\xi}(s_2)\tilde{\sigma}^-(s_1){\rho}_A^{\alpha}(0)\tilde{\sigma}^+(s_2) + \frac{\gamma^2}{4}\int_0^t ds_1\int_0^t ds'_1 \tilde{P}_1(s_1){\rho}_A^{\alpha}(0)\tilde{P}_1(s'_1) \\ \nonumber
&- \frac{\gamma}{2}\int\mathcal{D}\xi\,e^{-\|\xi\|^2}\int_0^t ds_1\int_0^t ds'_1\int_0^{s'_1}ds'_2\Big(f^*_{\xi}(s_1)f_{\xi}(s'_2) \tilde{\sigma}^-(s_1){\rho}_A^{\alpha}(0)\tilde{\sigma}^+(s'_2)\tilde{P}_1(s'_1) \\ \nonumber
&+f^*_{\xi}(s_1)f_{\xi}(s'_2)\tilde{P}_1(s'_1)\tilde{\sigma}^-(s'_2){\rho}_A^{\alpha}(0)\tilde{\sigma}^+(s_1)+f^*_{\xi}(s_1)f_{\xi}(s'_1)\tilde{\sigma}^-(s_1){\rho}_A^{\alpha}(0)\tilde{P}_1(s'_2)\tilde{\sigma}^+(s'_1)\\ \nonumber
&+f^*_{\xi}(s_1)f_{\xi}(s'_1)\tilde{\sigma}^-(s'_1)\tilde{P}_1(s'_2){\rho}_A^{\alpha}(0)\tilde{\sigma}^+(s_1)\Big)
+ \frac{\gamma^2}{4}\int_0^t ds_1\int_0^{s_1} ds'_1\Big({\rho}_A^{\alpha}(0)\tilde{P}_1(s'_1)\tilde{P}_1(s_1)+ \tilde{P}_1(s_1)\tilde{P}_1(s'_1){\rho}_A^{\alpha}(0)\Big)\\\nonumber
&+ \int\mathcal{D}\xi\ e^{-\|\xi\|^2} \int_0^t ds_1\int_0^{s_1} ds'_1\int_0^t ds_2\int_0^{s_2} ds'_2 f^*_{\xi}(s_1)f_{\xi}(s_2)f^*_{\xi}(s'_1)f_{\xi}(s'_2)\tilde{\sigma}^-(s_1)\tilde{\sigma}^-(s'_1) {\rho}_A^{\alpha}(0)\tilde{\sigma}^+(s'_2)\tilde{\sigma}^+(s_2)\Big]U_{\alpha}^\dagger(t,0).\numberthis
\end{align*}
Now applying the approximation for $F^r(t-s)$ carefully we get
\begin{align*}\label{nv8}
\rho_A(t)&= \Lambda (t)\rho_A(0)
\approx U_{\alpha}(t,0)\Big[\rho_A(0) -\int_0^t ds_1\frac{\gamma}{2}\left(\tilde{P}_1(s_1)\rho_A(0)+\rho_A(0)\tilde{P}_1(s_1)\right)+ \gamma\int_0^t ds_1\  \tilde{\sigma}^-(s_1)\rho_A(0)\tilde{\sigma}^+(s_1)\\\nonumber
&+ \frac{\gamma^2}{4}\int_0^t ds_1\int_0^t ds'_1 \tilde{P}_1(s_1)\rho_A(0)\tilde{P}_1(s'_1)
- \frac{\gamma^2}{2}\int_0^t ds'_1\int_0^{s'_1} ds'_2 \Big(\tilde{\sigma}^-(s'_2)\rho_A(0)\tilde{\sigma}^+(s'_2)\tilde{P}_1(s'_1)+ \tilde{P}_1(s'_1)\tilde{\sigma}^-(s'_2)\rho_A(0)\tilde{\sigma}^+(s'_2)\Big)\\\nonumber
&- \frac{\gamma^2}{2}\int_0^t ds_1\int_0^{s_1} ds'_2\ \Big(\tilde{\sigma}^-(s_1)\rho_A(0)\tilde{P}_1(s'_2)\tilde{\sigma}^+(s_1)+ \tilde{\sigma}^-(s_1)\tilde{P}_1(s'_2)\rho_A(0)\tilde{\sigma}^+(s_1)\Big)+ \frac{\gamma^2}{4}\int_0^t ds_1\int_0^t ds'_1\Big(\rho_A(0)\tilde{P}_1(s'_1)\tilde{P}_1(s_1)\\\nonumber
&+ \tilde{P}_1(s_1)\tilde{P}_1(s'_1)\rho_A(0)\Big)+ \gamma^2\int_0^t ds_1\int_0^{s_1} ds'_1\ \tilde{\sigma}^-(s_1)\tilde{\sigma}^-(s'_1)\rho_A(0)
\tilde{\sigma}^+(s_1)\tilde{\sigma}^+(s'_1)\Big]U_{\alpha}^\dagger(t,0).\numberthis
\end{align*}
Now writing the solution of Eq. \eqref{nv51} to the second order in Dyson series, and going back to Schr\"{o}dinger picture, we obtain 
\begin{align*}\label{nv9}
     \Lambda (t)(\cdot)&= \mathcal{U}_{\alpha}(t,0)\Big[\mathbb{I}-\frac{\gamma}{2}\int_0^t ds_1\{\tilde{P}_1(s_1), (\cdot)\}
     + \gamma\int_0^t ds_1\  \tilde{\sigma}^-(s_1)(\cdot)\tilde{\sigma}^+(s_1)+ \frac{\gamma^2}{4}\int_0^t ds_2\Big\{\tilde{P}_1(s_2), \int_0^{s_2} ds_1\{\tilde{P}_1(s_1), (\cdot)\}\Big \}\\\nonumber 
     &- \frac{\gamma^2}{2}\int_0^t ds_2 \int_0^{s_2} ds_1\  \Big\{\tilde{P}_1(s_2),\tilde{\sigma}^-(s_1)(\cdot)\tilde{\sigma}^+(s_1)\Big\}- \frac{\gamma^2}{2}\int_0^t ds_2 \int_0^{s_2} ds_1\  \tilde{\sigma}^-(s_2)\{\tilde{P}_1(s_1), (\cdot)\}\tilde{\sigma}^+(s_2)\\\nonumber
     &+ \gamma^2 \int_0^t ds_1 \int_0^{s_1} ds_2\  \tilde{\sigma}^-(s_1)\tilde{\sigma}^-(s_2)(\cdot)\tilde{\sigma}^+(s_2)\tilde{\sigma}^+(s_1)
     \Big] + \mathcal{O}(\gamma^3).\numberthis
\end{align*}
Hence for the evolution of the state of the atom, in the Schr\"{o}dinger picture, we have
\begin{align*}\label{nv10}
     \rho_A(t)&= \Lambda (t){\rho}_A^{\alpha}(0)= \mathcal{U}_{\alpha}(t,0)\Big[\mathbb{I}-\frac{\gamma}{2}\int_0^t ds_1\{\tilde{P}_1(s_1), {\rho}_A^{\alpha}(0)\}+ \gamma\int_0^t ds_1\  \tilde{\sigma}^-(s_1){\rho}_A^{\alpha}(0)\tilde{\sigma}^+(s_1)\\ \nonumber
     &+  \frac{\gamma^2}{4}\int_0^t ds_2\int_0^{s_2} ds_1\Big(\tilde{P}_1(s_2)\tilde{P}_1(s_1){\rho}_A^{\alpha}(0)+{\rho}_A^{\alpha}(0)\tilde{P}_1(s_1)\tilde{P}_1(s_2)\Big)+ \frac{\gamma^2}{4} 2 \int_0^t ds_2\int_0^{s_2} ds_1\tilde{P}_1(s_1){\rho}_A^{\alpha}(0)\tilde{P}_1(s_2)\\\nonumber
     &-  \frac{\gamma^2}{2}\int_0^t ds_2\int_0^{s_2} ds_1\ \Big(\tilde{P}_1(s_2)\tilde{\sigma}^-(s_1){\rho}_A^{\alpha}(0)\tilde{\sigma}^+(s_1)+\tilde{\sigma}^-(s_1){\rho}_A^{\alpha}(0)\tilde{\sigma}^+(s_1)\tilde{P}_1(s_2)\Big)\\\nonumber
     &- \frac{\gamma^2}{2}\int_0^t ds_2\int_0^{s_2} ds_1\ \Big(\tilde{\sigma}^-(s_2)\tilde{P}_1(s_1){\rho}_A^{\alpha}(0)\tilde{\sigma}^+(s_2)+\tilde{\sigma}^-(s_2){\rho}_A^{\alpha}(0)\tilde{P}_1(s_1)\tilde{\sigma}^+(s_2)\Big)\\\nonumber
     &+ \gamma^2\int_0^t ds_2\int_0^{s_2} ds_1\  \tilde{\sigma}^-(s_2)\tilde{\sigma}^-(s_1){\rho}_A^{\alpha}(0)\tilde{\sigma}^+(s_1) \tilde{\sigma}^+(s_2)\Big] + \mathcal{O}(\gamma^3).\numberthis
\end{align*}

Then comparing Eqs. \eqref{nv8} and \eqref{nv10} we see that they are the same. The only term which is non-trivial to see it is 
\begin{equation}
    \frac{\gamma^2}{4} 2 \int_0^t ds_2\int_0^{s_2} ds_1\tilde{P}_1(s_1){\rho}_A^{\alpha}(0)\tilde{P}_1(s_2)
\end{equation}
but checking the integration limits, we can see that $2 \int_0^t ds_2\int_0^{s_2} ds_1\tilde{P}_1(s_1){\rho}_A^{\alpha}(0)\tilde{P}_1(s_2)$ is equal to $\int_0^t ds_2\int_0^{t} ds_1\tilde{P}_1(s_1){\rho}_A^{\alpha}(0)\tilde{P}_1(s_2)$ as we want. Therefore, up to second order in $\gamma$, we see that the evolution is given by the map $\mathcal{T} e^{\tilde{\Lambda}(t)}$ where
\begin{equation}
   \tilde{\Lambda}(t)=-\frac{\gamma}{2}\int_0^t ds_1\{\tilde{P}_1(s_1), (\cdot)\}
     + \gamma\int_0^t ds_1\ \tilde{\sigma}^-(s_1)(\cdot)\tilde{\sigma}^+(s_1).
\end{equation}

\end{document}